\documentclass[letterpaper,11pt]{article}

\usepackage{amsmath}
\usepackage{amsfonts}
\usepackage{amssymb}
\usepackage[margin=1in]{geometry}
\usepackage{mathtools}

\usepackage[T1]{fontenc}
\usepackage[utf8]{inputenc}
\usepackage{microtype}
\usepackage{booktabs}
\usepackage{subcaption}
\usepackage[ruled]{algorithm}
\usepackage{algpseudocodex}
\AtEndPreamble{\usepackage[capitalize]{cleveref}}
\usepackage{amsthm}
\usepackage{thmtools,thm-restate}
\usepackage[hidelinks]{hyperref}
\usepackage{makecell}
\hypersetup{hypertexnames=false}
\declaretheorem[numberwithin=section]{theorem}
\declaretheorem[style=plain,sharenumber=theorem]{lemma}
\declaretheorem[style=plain,sharenumber=theorem]{corollary}
\declaretheorem[style=plain,sharenumber=theorem]{proposition}

\declaretheorem[style=definition,sharenumber=theorem]{definition}
\declaretheorem[style=definition,sharenumber=theorem]{example}

\declaretheorem[style=remark,sharenumber=theorem]{remark}

\title{On the Impact of Stability and the Helly Property on the Dominating Set Problem}

\author{Che Cheng\\
\small{National Taiwan University}\\
\small{\texttt{\href{mailto:che.cheng@rwth-aachen.de}{che.cheng@rwth-aachen.de}}}
\and 
Daniel Mock\\
\small{RWTH Aachen University}\\
\small{\texttt{\href{mailto:mock@cs.rwth-aachen.de}{mock@cs.rwth-aachen.de}}}
\and 
Peter Rossmanith\\
\small{RWTH Aachen University}
\\
\small{\texttt{\href{mailto:rossmani@cs.rwth-aachen.de}{rossmani@cs.rwth-aachen.de}}}
}

\date{}

\definecolor{lipicsGray}{RGB}{0.31,0.31,0.33}
\definecolor{lipicsLineGray}{RGB}{0.51,0.50,0.52}

\usepackage{xspace}
\usepackage{todonotes}
\usepackage{listings}
\usepackage{framed}
\usepackage{xspace}
\usepackage{tikz}
\usetikzlibrary{calc}
\usetikzlibrary{shapes}
\usetikzlibrary{arrows.meta}

\usepackage{xifthen}
\usepackage{tabularx}
\usepackage[most]{tcolorbox}

\newtcolorbox{GrayBox}[1]{
enhanced,
breakable,
arc=2mm,
boxrule=0.5pt,
colframe=lipicsLineGray,
colback=white,
width=0.97\textwidth,
center,
left=2ex,
right=2ex,
title={#1},
coltitle=black,
attach boxed title to top left={
yshift=-\tcboxedtitleheight/2,
xshift=6pt
},
boxed title style={
colback=white,
colframe=white,
outer arc=0mm,
arc=0mm,
top=0pt, bottom=0pt, left=0pt, right=0pt
},
}

\newcommand{\C}{\mathcal{C}}
\newcommand{\F}{\mathcal{F}}
\newcommand{\N}{\mathbf{N}}
\newcommand{\Z}{\mathbf{Z}}
\renewcommand{\phi}{\varphi}
\renewcommand{\epsilon}{\varepsilon}
\newcommand{\vw}{\bar{w}}
\newcommand{\vx}{\bar{x}}
\newcommand{\vy}{\bar{y}}

\newcommand{\va}{\bar{a}}

\newcommand{\vc}{\bar{c}}

\DeclareMathOperator{\mw}{mw}

\DeclareMathOperator{\dist}{dist}

\DeclareMathOperator{\SCd}{SCd}
\newcommand{\W}[1]{\mathrm{W}[#1]}
\newcommand{\dom}{\mathrm{dom}}

\providecommand\given{}
\newcommand\SetSymbol[1][]{
\nonscript\,#1\vert \allowbreak \nonscript\,\mathopen{}}
\DeclarePairedDelimiterX\Set[1]{\lbrace}{\rbrace}%
{ \renewcommand\given{\SetSymbol[\delimsize]} #1 }
\DeclarePairedDelimiterX{\Paren}[1]{(}{)}{#1}
\DeclarePairedDelimiterX{\Abs}[1]{\lvert}{\rvert}{#1}
\DeclarePairedDelimiterX{\Norm}[1]{\lVert}{\rVert}{#1}
\DeclarePairedDelimiterX{\Ceil}[1]{\lceil}{\rceil}{#1}
\DeclarePairedDelimiterX{\Floor}[1]{\lfloor}{\rfloor}{#1}

\newcolumntype{R}{>{\displaystyle}r}
\newcolumntype{C}{>{\displaystyle}c}
\newcolumntype{L}{>{\displaystyle}l}

\DeclareFontShape{T1}{cmr}{m}{scit}{<-> ssub*cmr/m/scsl}{}

\begin{document}

\begin{titlepage}
    \maketitle

    \begin{abstract}
        We extend the algorithmic framework of \emph{progressive exploration}
        [Fabiański et al., \emph{STACS 2019}], which yields simple,
        yet surprisingly general and efficient parameterized algorithms
        for \textsc{Dominating Set},
        \textsc{Independent Set},
        and some of their variants.
        While they identified
        \emph{stability} and the \emph{Helly property} as necessary for their
        approach, we show that---with a simple change---in the case of
        \textsc{Dominating Set}, one can get rid of the stability requirement. This
        yields a fixed-parameter tractable algorithm
        on exactly those graph classes which do not contain long \emph{co-matchings} or \emph{double-ladders} as semi-induced subgraphs.
        Lifting one of these two restrictions makes \textsc{Dominating Set} W[1]-hard on these classes.
        Our algorithm generalizes results on
        weakly $\gamma$-closed graphs,
        and results from Sparsity theory, e.g.,
        nowhere dense and biclique-free classes.
        At the same time, we match the time
        complexity of the previously known algorithms on those classes.
        We demonstrate that this technique can easily be applied to the \textsc{Distance-$r$ Dominating Set} and the \textsc{Set Cover} problem.
    \end{abstract}

    \paragraph*{Acknowledgements}
    We want to thank Nikolas Mählmann and Sebastian Siebertz for pointing us to the hardness results for ladder-free classes and co-matching-free classes.
    Daniel Mock and Peter Rossmanith were funded by the Deutsche Forschungsgemeinschaft (DFG, German Science Foundation) -- DFG-927/15-2.
    Che Cheng was funded by the National Science and Technology Council of Taiwan (NSTC) -- 114-2221-E-002-183-MY3.
    \thispagestyle{empty}
\end{titlepage}

\section{Introduction}\label{sec:intro}

The
\textsc{Dominating Set} problem is one of the fundamental problems in classical
and parameterized complexity. The problem is $\W{2}$-hard on general graphs when
parameterized by solution size $k$, which motivates the question of which
restricted graph classes allow for fixed-parameter tractability of this problem.
This led to a plethora of both positive and negative tractability results and
techniques for fixed-parameter tractable (short: fpt)
algorithms and kernelization.
Currently, some of the most general graph classes where fixed-parameter
tractability of \textsc{Dominating Set} is known are biclique-free
classes~\cite{PhilipRS12}, weakly $\gamma$-closed
classes~\cite{LokshtanovS21},
monadically stable classes~\cite{DreierMS23,DreierEMMPT24},
and classes of bounded merge-width~\cite{DreierT25}.%
\footnote{Assuming a construction sequence witnessing the merge-width is given with the input.}

To unify and generalize some of the above results while investigating their
relationship to the others, the authors in~\cite{FabianskiPST19} proposed
the \emph{progressive exploration} framework,
which applies, among other problems, to \textsc{Dominating Set}.
To solve the \textsc{Dominating Set} problem, the so-called
\emph{Semi-Ladder Algorithm} performs iteratively in each round two steps:
first, it searches
for a candidate solution $D$ of size $k$ that dominates all previously seen
\emph{``witnesses''}; second, it searches for a witness,
i.e., a vertex that is not
dominated by the candidate solution $D$. If no candidate solution could be found
in the first step, the algorithm terminates with the answer \textsc{no}. If no
witness can be found in the second step, $D$ dominates all vertices and the
algorithm terminates with the answer \textsc{yes}. Otherwise, it proceeds with
the next round.

As it is, the algorithm seems to be nothing more than a brute-force search with
a running time of $O(n^k)$.
However, if the number of rounds, and hence the size of the set of
witnesses, is bounded, the running time of this algorithm is fpt. Surprisingly,
for many graph classes the number of rounds can be bounded by a function of $k$
alone. Especially, this reestablishes the fixed-parameter tractability of the
\textsc{Dominating Set} problem for all biclique-free classes and all powers of
nowhere-dense classes while improving or matching the best known running times
for these classes~\cite{PhilipRS12,DawarK09}.

To analyze the running time,
the key observation is that the algorithm implicitly
constructs a \emph{semi-ladder} (see \cref{fig:semi-ladder}) of candidate
solutions and witnesses during its run inside the
\emph{agreement graph,}%
\footnote{Referred to as the \emph{auxiliary graph} in \cite{FabianskiPST19}.}
a central notion throughout their and our analysis.
The
agreement graph is a bipartite graph consisting of the candidate solutions,
which are $k$-tuples of vertices of $G$, on the one side and witnesses,
which are single vertices of
$G$, on the other. A candidate solution and a witness are adjacent if the
candidate solution dominates the witness.
Instead of constructing this graph explicitly,
the algorithm only ``explores''
parts of the agreement graph along a semi-ladder.
The length of the longest
semi-ladder in the agreement graph determines the maximum number of rounds the
algorithm executes and hence, its running time.

This length, called the \emph{semi-ladder index} of the agreement graph,
is bounded (by a function of $k$)
if and only if the \emph{co-matching index} and the \emph{ladder index} of the
agreement graph are bounded. These are the lengths of the longest co-matchings
and ladders, and they correspond to the \emph{Helly property} and
Shelah's \emph{stability} from model theory, respectively.
For an illustration of these structures, see \cref{fig:co-matching,fig:ladder}.
To show the fpt running of the Semi-Ladder Algorithm,
it suffices to bound the co-matching and ladders indices of the graph,
which is done for biclique-free classes and powers of nowhere dense classes in
\cite{FabianskiPST19}.

\begin{figure}[t]
    \centering
    \begin{subfigure}[b]{.5\textwidth}
        \centering
        \includegraphics[height=3.2cm]{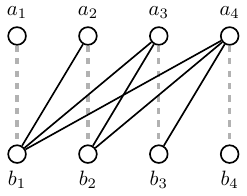}
        \caption{Semi-ladder}
        \label{fig:semi-ladder}
    \end{subfigure}%
    \begin{subfigure}[b]{.5\textwidth}
        \centering
        \includegraphics[height=3.2cm]{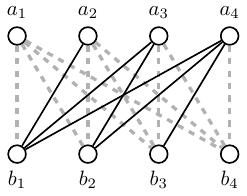}
        \caption{Ladder}
        \label{fig:ladder}
    \end{subfigure}%
    \medskip

    \begin{subfigure}[b]{.5\textwidth}
        \centering
        \includegraphics[height=3.2cm]{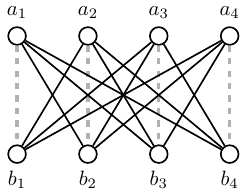}
        \caption{Co-matching}
        \label{fig:co-matching}
    \end{subfigure}%
    \begin{subfigure}[b]{.5\textwidth}
        \centering
        \includegraphics[height=3.2cm]{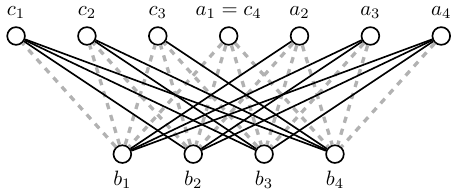}
        \caption{Double-ladder}
        \label{fig:double-ladder}
    \end{subfigure}%
    \caption{A semi-ladder, a ladder, a co-matching, and a double-ladder of
        order $4$, respectively.
        Solid lines represent edges and dashed lines represent non-edges.
        If no line is drawn between two vertices,
        then their connection can be arbitrary.}
    \label{fig:obstructions}
\end{figure}

As the tractability of the Semi-Ladder Algorithm relies on both the Helly
property and stability, we investigate whether one can generalize the framework
to work only with one of these properties.

\subsection{New Results}
We extend the progressive exploration framework to work with only the Helly
property via the Co-Matching Algorithm,
and show that the resulting algorithm is fpt for
\textsc{Dominating Set} on exactly the classes
forbidding large co-matchings and double-ladders
(\cref{fig:co-matching,fig:double-ladder}),
thereby unifying and extending semi-ladder-free and
weakly $\gamma$-closed classes.
The former is the class on which the Semi-Ladder Algorithm achieves fpt
running times,
and the latter originates from the triadic closure property in social
networks and generalizes classes of bounded degeneracy.
We also show how to adapt the framework to \textsc{Partial Dominating Set},
\textsc{Set Cover}, and \textsc{Distance-$r$ Dominating Set}.
A more detailed description of each contribution follows.

\paragraph*{Contribution 1: Algorithm for Helly Property}
The main contribution of this work is an algorithmic extension of the
progressive exploration framework. In particular, we show that the progressive
exploration framework can be made to work with only the Helly property by
introducing the \emph{Co-Matching Algorithm} for the \textsc{Dominating Set}
problem (\cref{sec:co-matching-algo}).
The algorithm extends the Semi-Ladder Algorithm by adding a
\emph{Compression Step} that reduces the size of the witness set before
searching for a new candidate.
This step removes ``redundant'' vertices from the witness set---that is,
those vertices that do not contribute to restricting the set of candidates that
dominate the witness set.
This allows the algorithm to run an unbounded
number of rounds while keeping the oracle calls efficient, and as a result
removes the need for stability.

\paragraph*{Contribution 2: Characterization of Graph Classes with Helly Property}
We further explore the graph classes on which the Co-Matching Algorithm achieves
an fpt running time for \textsc{Dominating Set},
and show that these classes can be characterized exactly by
graph classes forbidding arbitrarily large co-matchings and double-ladders
(\cref{fig:co-matching,fig:double-ladder}) as semi-induced subgraphs
(\cref{sec:graph-class}). The purely
structural characterization allows us to make connections to other well-studied
graph classes easily. The characterization is visualized in \cref{fig:classes},
where the highlighted green box (bd.\@ $\dom^k$-co-matching index) corresponds to the classes
where the Co-Matching Algorithm solves \textsc{Dominating Set} in fpt time. In
particular, we show that it unifies and extends both semi-ladder-free graph
classes~\cite{FabianskiPST19}
and weakly $\gamma$-closed graph classes~\cite{LokshtanovS21}.
The running time of the Co-Matching Algorithm on these classes matches the
previously known results.
On the negative side,
we show that classes of bounded shrubdepth
do not have this property.
Hence, while \textsc{Dominating Set} is fpt on these classes,
the Co-Matching Algorithm is not.

\begin{figure}[t]
    \centering
    \includegraphics[width=\textwidth]{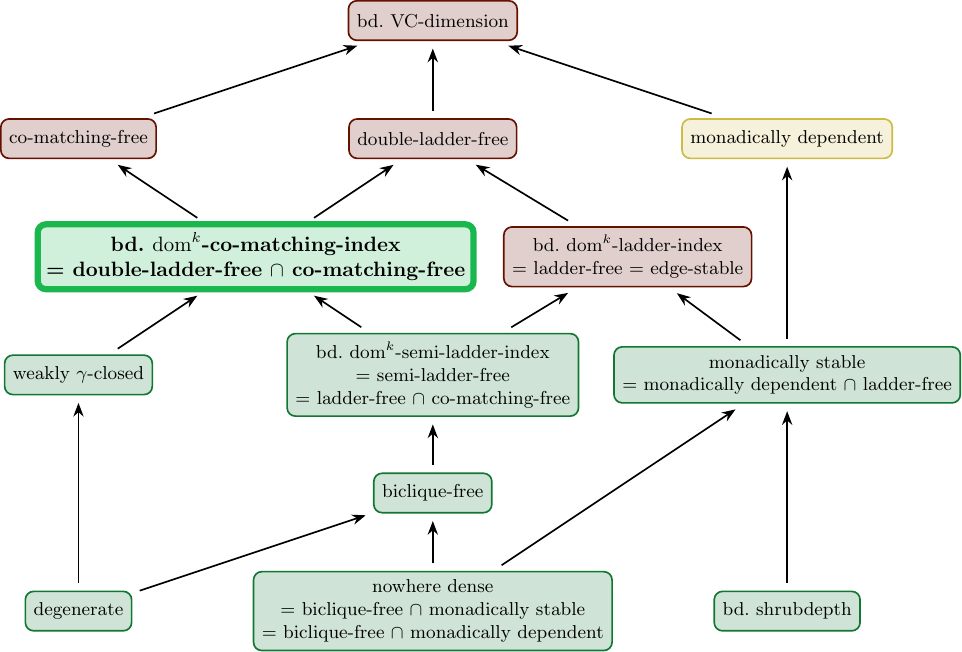}
    \caption{Inclusion structure of the graph classes discussed in this paper.
        The colors denote (in)tractability of \textsc{Dominating Set}
        parameterized by the solution size $k$:
        fpt algorithms exist for green classes;
        \textsc{Dominating Set} is $\W{1}$-hard on red classes;
        and whether it is fpt on monadically dependent classes is unknown but conjectured to be fpt~\cite{DreierMS23}.
        The highlighted class shows our result.}
    \label{fig:classes}
\end{figure}

\paragraph*{Observation: Hardness for Ladder-Free Classes and for Co-Matching-Free Classes}
Our results are tight in the following sense: If one drops the condition of
double-ladder-freeness or co-matching-freeness, i.e., only forbids co-matchings
or only forbids double-ladders as semi-induced subgraphs, then there exists such a class
where \textsc{Dominating Set} is $\W{1}$-hard.
We also note the hardness on ladder-free classes (\cref{prop:domset-hardness}).
These classes correspond to the agreement graphs having bounded ladder index, implying that the stability property alone does not ensure tractability for \textsc{Dominating Set}.
This answers a question raised by Guillemot~\cite{Guillemot25} (for hypergraphs) in the negative.
These hardness results are known from literature and are shaded red in \cref{fig:classes}.

\paragraph*{Contribution 3: Application to Dominating Set Variants}

We discuss that with straightforward modifications, the Co-Matching Algorithm can solve \textsc{Distance-$r$ Dominating Set} (\cref{sec:dist-r-ds}),
domination-type problems (\cref{sec:domination-type}),
and \textsc{Set Cover} (\cref{sec:set-cover}).
For the distance-$r$ variant,
the classes where the algorithm is fpt for \textsc{Distance-$r$ Dominating Set}
subsume those
where the Semi-Ladder Algorithm is fpt.
However,
we only have an indirect characterization of those classes and are not
aware of any natural classes of this sort beyond the capability of the
Semi-Ladder Algorithm.
For domination-type problems,
we show that while the Co-Matching Algorithm is fpt on
co-matching-free and double-ladder-free classes
for a non-trivial subset of such problems,
in some cases it only achieves fpt running time for semi-ladder-free classes.
For \textsc{Set Cover},
by defining the notions in a straightforward sense,
we show that the Co-Matching Algorithm is fpt for \textsc{Set Cover} on
co-matching-free and double-ladder-free set systems,
and that similar extension to \emph{coverage-type problems} is possible.
We also apply the progressive exploration framework to the
\textsc{Partial Dominating Set} problem,
a variant of \textsc{Dominating Set} where not all but a given number
$t\leq \Abs{G}$ of vertices have to be dominated (\cref{sec:pds}).
We adapt the Semi-Ladder Algorithm to obtain a randomized fpt approximation algorithm for this problem on semi-ladder free classes.
While this result is subsumed by existing literature~\cite{Guillemot25},
it demonstrates the flexibility of the progressive exploration framework for variants of the \textsc{Dominating Set} problem.

\subsection{Related Work}

\paragraph*{Historical Overview of \textsc{Dominating Set} Tractability}

The tractability landscape of \textsc{Dominating Set} can be organized into
two broad approaches. Table~\ref{tab:history} gives a chronological
overview. The first approach consists of direct
combinatorial algorithms tailored to \textsc{Dominating Set} that achieve
elementary fpt running times on classes defined by concrete structural
restrictions---beginning with bounded treewidth and planarity, progressing
through domination cores on nowhere dense classes
and weakly $\gamma$-closed classes and the progressive
exploration framework on biclique-free classes,
and culminating in the Co-Matching Algorithm on
co-matching-and-double-ladder-free classes (this work).
An outlier in this approach is that of claw-free graphs~\cite{HermelinMLW11},
i.e.,
graphs that do not contain $K_{1,3}$ as induced subgraph.
The class has unbounded VC-dimension,
which puts it outside of all other classes considered in the paper.
A decomposition unrelated to any other classes is used to derive the fpt result.
The second approach
comprises general first-order model-checking algorithms that solve any
FO-definable problem, including \textsc{Dominating Set}, on extremely
general classes such as monadically stable graphs and classes of bounded
merge-width;%
\footnote{Assuming a construction sequence witnessing the merge-width is given
    with the input.}
the cost is heavy machinery and non-elementary dependence on
the parameter.

\begin{table}[t]
    \centering
    \setlength{\tabcolsep}{5pt}
    \begin{tabular}{clll}
        \toprule
        Year       & Graph class                                                  & Reference & Comment \\
        \midrule
        $\sim$1999 & \llap{\textsuperscript{*}}Bounded treewidth
                   & Downey, Fellows~\cite{DowneyFellows13}
                   & DP on tree decompositions                                                          \\
        2002       & \llap{\textsuperscript{*}}Planar graphs
                   & Alber, Bodlaender et al.~\cite{AlberBFKN02}
                   & Subexponential                                                                     \\
        2004       & \llap{\textsuperscript{*}}Bounded genus
                   & Ellis, Fan, Fellows~\cite{EllisFF04}                                               \\
        2005       & \llap{\textsuperscript{*}}$H$-minor-free
                   & Demaine, Fomin et al.~\cite{DemaineFHT05}
                   & Bidimensionality                                                                   \\
        2009       & \llap{\textsuperscript{*}}Degenerate graphs
                   & Alon, Gutner~\cite{AlonG09}                                                        \\
        2009       & \llap{\textsuperscript{*}}Nowhere dense
                   & Dawar, Kreutzer~\cite{DawarK09}
                   & Domination cores                                                                   \\
        2009       & \llap{\textsuperscript{*}}$K_h$-top.-minor-free
                   & Alon, Gutner~\cite{AlonGutner09}                                                   \\
        2009       & \llap{\textsuperscript{*}}$K_{i,j}$-free
                   & Philip, Raman, Sikdar~\cite{PhilipRS09}                                            \\
        2011       & \llap{\textsuperscript{*}}Claw-free
                   & Hermelin, Mnich et al.~\cite{HermelinMLW11}
                   & Unbounded VC-dimension                                                             \\
        2019       & \llap{\textsuperscript{*}}Semi-ladder-free
                   & Fabia\'nski, Pilipczuk et al.~\cite{FabianskiPST19}
                   & Progressive exploration                                                            \\
        2020       & \llap{\textsuperscript{\textdagger}}Bounded twin-width
                   & Bonnet, Kim et al.~\cite{BonnetKTW20}
                   & FO model checking                                                                  \\
        2021       & \llap{\textsuperscript{*}}Weakly $\gamma$-closed
                   & Lokshtanov, Surianarayanan~\cite{LokshtanovS21}
                   & Domination cores                                                                   \\
        2023       & Bounded shrubdepth
                   & Bergougnoux, Chekan et al.~\cite{BergougnouxCGKMOPL23}                             \\
        2024       & Monadically stable
                   & Dreier, Eleftheriadis et al.~\cite{DreierEMMPT24,DreierMS23}
                   & FO model checking                                                                  \\
        2025       & \llap{\textsuperscript{\textdagger}}Bounded merge-width
                   & Dreier, Toru\'nczyk~\cite{DreierT25}
                   & FO model checking                                                                  \\
        \midrule
        2026       & \makecell{double-ladder-free                                                       \\$\cap$\ co-matching-free}
                   & This paper
                   & Co-Matching Algorithm                                                              \\
        \bottomrule
    \end{tabular}
    \caption{Progress of fixed-parameter tractability of
        \textsc{Dominating Set}.
        Starred classes are subsumed by our class.
        For classes of bounded twin-width/merge-width (marked by \textdagger),
        a contraction/construction sequence is assumed to be given as part of
        the input.}
    \label{tab:history}
\end{table}

\paragraph*{Connection to Domination Cores}
A \emph{$k$-domination core} is a set of vertices such that every set of size
at most $k$ that dominates it dominates the whole graph.
This notion was first introduced in \cite{DawarK09},
and later used as a kernelization technique for classes of bounded expansion and
nowhere dense classes~\cite{DrangeDFKL16,EickmeyerGKKPRS17,FabianskiPST19} and
to derive an fpt algorithm for weakly $\gamma$-closed
graphs~\cite{LokshtanovS21}.
These techniques
typically start with the set of all vertices, and repeatedly shrink the size of
the core by identifying and removing redundant vertices. A compact presentation
of the notion can be found in \cite[Chapter~4]{NesetrilM12}. As already
identified in \cite{FabianskiPST19}, the progressive exploration framework has a
natural connection to domination cores, but constructs them from the opposite
direction---%
an initially empty set of witnesses grows iteratively until it forms a
domination core, at which point the algorithm terminates. Our approach is a
hybrid of the two: The core is constructed in a bottom-up fashion while being
shrunk at each step.
We observe that the co-matching index of the agreement graph
(the maximum order of a co-matching it contains)
bounds the maximum size of a minimal domination core.
That is, the classes on which the Co-Matching Algorithm achieves an fpt
running time are exactly those where the size of all minimal domination cores
are bounded.

\paragraph*{Connection to First-Order Model Checking}
Given a first-order (FO) formula $\phi$ and a graph $G$,
the \textsc{FO Model Checking} problem asks whether $G\models \phi$.
The known tractability
frontier of such model-checking algorithms is marked by monadically stable
classes~\cite{DreierMS23,DreierEMMPT24} and classes of bounded
merge-width~\cite{DreierT25}.%
\footnote{Assuming a construction sequence witnessing the merge-width is given with the input.}
Since \textsc{Dominating Set} is FO-definable, it can be solved using such
algorithms. However, these results are notorious for their non-elementary
dependence on the parameters and provide no combinatorial insights.
Moreover,
it is proven that FO
model checking is $\W{1}$-hard on all hereditary and monadically independent
graph classes~\cite{DreierMT24}.
While our approach has strong connections to
model theory concepts such as stability and the Helly property, by focusing only
on \textsc{Dominating Set}, our result is simple to understand, has a
single-exponential parameter dependence, and applies to graph classes beyond the
FO-tractability barrier by subsuming hereditary and monadically independent
classes such as biclique-free graphs and weakly $\gamma$-closed graphs
(see \cref{fig:classes}).

\paragraph*{Known Hardness Results}
{\sc Dominating Set} is $\W{1}$-hard on unit disk graphs, unit square
graphs~\cite{Marx06}, and on $K_{1,4}$-induced-free graphs~\cite{HermelinMLW11}.
Moreover, bounded VC-dimension does not seem to be a dividing line for
\textsc{Dominating Set}, as there exist classes whose VC-dimension and dual
VC-dimension is~2 where \textsc{Dominating Set} is $\W{1}$-hard and hereditary
classes of unbounded VC-dimension where \textsc{Dominating Set} is fpt, e.g.,
claw-free graphs~\cite{HermelinMLW11}.
Recently,
the authors in \cite{DreierMS26} identified that \textsc{Dominating Set} is
$\W{1}$-hard on co-matching-free classes and on ladder-free classes by bounding
the co-matching/ladder index of known $\W{1}$-hard graph classes.

\subsection{Techniques}\label{sec:techniques}
We outline here the techniques used to derive our results.
We will write $\dom^k(G)$ to denote the agreement graph of $G$,
and the $\dom^k$-co-matching index of $G$ to denote the co-matching index
of $\dom^k(G)$,
i.e.,
the maximum order of a co-matching in $\dom^k(G)$.
Formal definitions will be given in \cref{sec:background}.

\paragraph*{Co-Matching Algorithm}
The \emph{Co-Matching Algorithm} solves the \textsc{Dominating Set} problem
with a running time bounded by the co-matching index of the agreement graph.
The algorithm maintains a set $W$ of witnesses (vertices that must be
dominated) and works in rounds, each consisting of three steps:
\begin{enumerate}
    \item \textbf{Candidate Step:} Given the current witness set $W$,
          find a candidate solution $\vc = \Set{c_1,\dots,c_k}$ that dominates
          all of $W$.
          If no such candidate exists, return \textsc{no}.
    \item \textbf{Witness Step:} Given a candidate $\vc$,
          search for a witness $w$ not dominated by $\vc$
          (i.e., $w \notin \bigcup_{i=1}^k N[c_i]$).
          If no such witness exists, $\vc$ is a dominating set and is returned.
    \item \textbf{Compression Step:} Add the new witness $w$ to $W$, then
          check whether some witness in $W$ is now \emph{redundant} (see below).
          If so, remove it and repeat this step.
          otherwise, keep all witnesses and continue with the next round.
\end{enumerate}

These steps rely on three oracles:
\begin{itemize}
    \item \textbf{Candidate Oracle:} Given a set of witnesses $W$,
          returns a candidate $(c_1,\dots,c_k)$ with $\bigcup_{i=1}^k N[c_i]
              \allowbreak \supseteq W$,
          or concludes that none exists.
    \item \textbf{Witness Oracle:} Given a candidate $(c_1,\dots,c_k)$,
          returns a vertex outside $\bigcup_{i=1}^k N[c_i]$,
          or concludes that all
          vertices are dominated.
    \item \textbf{Compression Oracle:} Given a set of witnesses $W$,
          returns a witness $w \in W$ that is redundant, or concludes that none exist.
\end{itemize}

The central innovation is the Compression Step, which keeps the witness set
small by repeatedly removing redundant witnesses. A witness $w$ is
\emph{redundant} for a set $W$ if every $k$-tuple of vertices that
dominates $W \setminus \Set{w}$ also dominates $w$
(\cref{def:redundant-witness}). In other words, removing $w$ from $W$ does
not enlarge the set of candidates that dominate $W$. irredundant witness
sets correspond to co-matchings in the agreement graph: if no witness in $W$
is redundant, then $W$ together with suitable candidates forms a co-matching
of order $\Abs{W}$ (\cref{cor:irredundant}). Hence, if the
co-matching index of the agreement graph is at most $\mu$, the Compression
Oracle ensures that $\Abs{W}\leq\mu+1$
throughout the run. The Co-Matching Algorithm
extends the Semi-Ladder Algorithm of \cite{FabianskiPST19}, which is fpt only
for classes of bounded semi-ladder index in the agreement graph.
Keeping the witness set small by removing redundant witnesses makes the
oracle calls efficient over an unbounded number of rounds
and removes the need for stability.

\begin{restatable*}{theorem}{ThmCoMatchingAlgoRuntime}
    \label{thm:co-matching-algo-runtime}
    The Co-Matching Algorithm solves \textnormal{\textsc{Dominating Set}}
    on a graph $G$ in time
    \[
        O\Paren*{\Norm{G} + \lambda\cdot (\mu+k)\cdot\Abs{G}+ \lambda\mu k\cdot\nu^G(\mu+1)^k}\,,
    \]
    where $\lambda$ is the \emph{$\dom^k$-semi-ladder index} of $G$,
    $\mu$ is the \emph{$\dom^k$-co-matching index} of $G$,
    and $\nu^G(m)$ is the \emph{neighborhood complexity}.
\end{restatable*}

For general graphs,
we have $\lambda\leq \Abs{G}$ and $\nu^G(m)\leq 2^m$,
and thus the algorithm is fpt if $\mu$ depends only on $k$.
A more refined running time that makes use of the characterization is given in
\cref{cor:runtime-class}.
The Co-Matching Algorithm matches the runtime of the Semi-Ladder Algorithm on any $dom^k$-semi-ladder-free class, especially on biclique-free and nowhere dense classes.
For those two classes, the runtime is then $\Norm{G} + k^{O(1)}\Abs{G} + k^{O(k)}$~(\cref{cor:runtime-nwd-biclique}).

\paragraph*{Characterization by Forbidden Semi-Induced Subgraphs}

To characterize the classes on which the Co-Matching Algorithm achieves
an fpt running time for \textsc{Dominating Set} (i.e., where $\mu$ is bounded),
we show that a long co-matching in the agreement graph $\dom^k(G)$
forces large forbidden substructures in $G$.
The idea is that,
for a candidate $\vc$ to dominate a witness $w$,
one of the $k$ vertices in $\vc$ has to dominate it.
We can thus color each unordered pair $\Set{i,j}$ of indices by a color $(p_1,p_2)$
such that $c^i_{p_1}$ (the $p_1$-th vertex in $\vc^i$) dominates $w_j$
and $c^j_{p_2}$ dominates $w_i$.
By Ramsey's theorem,
a sufficiently large co-matching in the agreement graph must contain
a large monochromatic subset.
The case where $p_1=p_2$ yields a co-matching,
whereas $p_1\neq p_2$ yields a double-ladder (\cref{lem:co-matching-Ramsey}).
We further show that we can construct these co-matchings and double-ladders as
semi-induced subgraphs by selecting a constant fraction of the indices which
guarantees that no vertex appears on both side of the structure
(\cref{lem:co-matching-semi-induced-delta_1,lem:double-ladder-semi-induced-delta_1}).
We thus get the following characterization.

\begin{restatable*}{theorem}{ThmCharacterizeCoMatchingIndex}
    \label{thm:characterize-co-matching-index}
    Let $\C$ be a graph class. For $k\geq 2$, $\dom^k(\C)$ has bounded
    co-matching index if and only if $\C$ is co-matching-free and
    double-ladder-free.
\end{restatable*}

We remark that for $k=1$,
the Ramsey argument has only one color and is therefore irrelevant,
and consequently,
we only need to forbid semi-induced co-matchings.

\Cref{thm:characterize-co-matching-index},
together with a bound on the VC-dimension of co-matching-free classes
(\cref{cor:delta_1-co-matching-VC-dimension}),
gives a more concrete bound on the running time of the Co-Matching Algorithm.

\begin{restatable*}{corollary}{CorRuntimeClass}
    \label{cor:runtime-class}
    Let $G$ be a graph of co-matching index and double-ladder index
    at most $\ell$.
    Then,
    the Co-Matching Algorithm solves the \textnormal{\textsc{Dominating Set}} problem
    on $G$ in time
    \[
        O\Paren*{k^{O(\ell k^2)}\Abs{G}^2 + k^{O(\ell^2 k^3)}\Abs{G}}\,.
    \]
\end{restatable*}

The purely
structural characterization also allows us to make connections to
other well-studied
graph classes easily. The relation is visualized in \cref{fig:classes},
where the highlighted green box (bd.\@ $\dom^k$-co-matching index)
corresponds to the classes
where the Co-Matching Algorithm solves \textsc{Dominating Set} in fpt time. In
particular, we show that it unifies and extends both semi-ladder-free graph
classes (i.e., the graph classes on which the Semi-Ladder Algorithm is fpt) and
\emph{weakly $\gamma$-closed} graph classes~\cite{LokshtanovS21}.
The connection to weakly $\gamma$-closed graphs can already be established
through observing that the authors in \cite{LokshtanovS21}
essentially bound the co-matching index in the agreement graph
(which corresponds to the maximum size of a minimal $k$-domination cores)
in order to show their fpt results (\cref{sec:threshold-set}).
Using our characterization,
we reprove and improve their bound
from $(\gamma-1)k\cdot R^{3^{15}k^2}(3\gamma)$ to $R^{k^2}(3\gamma)$.
Here,
$R^c(n)$ denotes the \emph{$c$-color Ramsey number}; Ramsey's theorem
guarantees $R^c(n)\leq c^{cn-1}$ for $c\geq 2$.

\begin{restatable*}{theorem}{ThmWeaklyClosedDomkCoMathing}\label{thm:weakly-closed-dom-k-co-matching}
    Every weakly $\gamma$-closed graph $G$ has
    $\dom^k$-co-matching index smaller than $R^{k^2}(3\gamma)$.
\end{restatable*}

Together with \cref{cor:runtime-class},
we also get a running time that matches the algorithm proposed in
\cite{LokshtanovS21} for \textsc{Dominating Set} on weakly
$\gamma$-closed graphs.

\begin{restatable*}{corollary}{CorWeaklyClosedRuntime}
    The Co-Matching Algorithm solves the \textnormal{\textsc{Dominating Set}} problem
    on weakly $\gamma$-closed graphs in time
    \[
        O\Paren*{k^{O(\gamma k^2)}\Abs{G}^2 + k^{O(\gamma^2 k^3)}\Abs{G}}\,.
    \]
\end{restatable*}

On the negative side,
we show two non-inclusion results that suggest the new tractability frontier
is incomparable with other tractability parameters for dense classes.
On the one hand,
we show that classes of bounded shrubdepth (and thus,
monadically stable classes and classes of bounded merge-width)
may contain arbitrarily large co-matchings.
Hence, the Co-Matching Algorithm is not fpt on these classes.

\begin{restatable*}{theorem}{ThmShrubdepthCoMatching}
    \label{thm:shrubdepth-co-matching}
    The class of co-matchings has bounded shrubdepth.
\end{restatable*}

On the other hand,
we show that $c$-closed graphs,
which generalize classes of bounded degree and have bounded
$\dom^k$-co-matching index,
have unbounded radius-$1$ merge-width.
This indicates that the tractability of such classes falls outside of the
merge-width regime.
\begin{restatable*}{theorem}{ThmWeaklyClosedMergeWidth}\label{thm:weakly-closed-merge-width}
    The class of all $c$-closed graphs has unbounded radius-$1$ merge-width. In
    particular, there exists a graph class that is biclique-free and $c$-closed
    that has unbounded degeneracy.
\end{restatable*}

\subsection{Organization of the Paper}
The rest of the paper is organized as follows.
\Cref{sec:background} establishes notation, formally defines the
agreement graph, and introduces the four obstructions---%
semi-ladders, ladders,
co-matchings, and double-ladders.
\Cref{sec:co-matching-algo} describes the Co-Matching Algorithm,
proves its correctness,
instantiates the three oracles for the
\textsc{Dominating Set} problem,
and derives the running time.
\Cref{sec:graph-class} characterizes the graph classes of bounded
$\dom^k$-co-matching index by forbidden semi-induced subgraphs, proves the
improved bound for weakly $\gamma$-closed graphs, and relates the new classes
to VC-dimension, merge-width, and shrubdepth.
\Cref{sec:beyond-domset} extends the framework to
\textsc{Distance-$r$ Dominating Set} and other domination-type problems,
explains why the Compression Step does
not help for \textsc{Independent Set}, adapts the Semi-Ladder Algorithm to
\textsc{Partial Dominating Set}, and generalizes both algorithms to
\textsc{Set Cover}.
\Cref{sec:conclusion} concludes with open questions.

\section{Preliminaries}
\label{sec:background}
We use the following notation. The set of all natural numbers is $\N$. For
$n\in\N$, we denote $\Set{1,\dots,n}$ by $[n]$. A graph $G$ has vertex set $V(G)$
and edge set $E(G)$. The order $\Abs{G}$ of a graph is $\Abs{V(G)}$ and the size
$\Norm{G}$ of a graph is $\Abs{E(G)}+\Abs{V(G)}$. A tuple of variables
$x^1,\dots,x^k$ is often abbreviated by $\vx$.

\paragraph*{Agreement Graph}
The \emph{agreement graph} is an intermediary structure which is used to analyze
the behavior of algorithms in the progressive exploration framework.
Given an FO formula of the form $\exists\vx\forall\vy\phi(\vx;\vy)$
and a graph $G$,
the agreement graph $\phi(G)$ intuitively captures the relationship between
candidate solutions to $\vx$ and their effects on $\vy$.
For the purpose of this paper, we focus on
the \textsc{Dominating Set} problem, which is defined by the formula
\[
    \dom^k(\vx; y)  \coloneq \bigvee_{i=1}^k \dom(x^i ; y) \text{, where } \dom(x ; y)\coloneq \dist(x,y) \leq 1\,.
\]

\begin{definition}[Agreement graph]\label{def:agreement-graph} Given $k\in\N$,
    the agreement graph $\dom^k(G)$ of a graph $G$ is a bipartite graph
    $(C^*,W^*,E)$, where
    \begin{itemize}
        \item $C^*$ is the set of all $k$-tuples of vertices in $G$, called the
              \emph{candidates,}
        \item $W^*$ is the set of all vertices in $G$, called the
              \emph{witnesses,} and
        \item an edge $\vc w \in E$ if and only if $G\models \dom^k(\vc; w)$.
              We say that $\vc$ and $w$ \emph{agree} in that case.
    \end{itemize}
    For a class $\C$ of graphs, we define
    $\dom^k(\C) = \Set{\dom^k(G) \given G\in \C}$.
    The agreement graph $\dom(G)$ is defined to be $\dom^1(G)$.
\end{definition}

Note that a set of size $k$ is a dominating set in $G$ if and only if it (or more formally, its corresponding candidate)
agrees with all witnesses in the
agreement graph $\dom^k(G)$.
While this problem is trivial to solve in time linear in $\Norm{\dom^k(G)}$,
constructing $\dom^k(G)$ explicitly takes $\Omega(\Abs{G}^k)$ time and space,
which is too large for an fpt algorithm. Hence,
algorithms in the progressive exploration framework only access the agreement
graph via certain oracles that can be efficiently implemented for certain graph
classes $\C$.
For instance, the \emph{Semi-Ladder Algorithm}~\cite{FabianskiPST19} has access
to the \emph{Candidate Oracle} and the \emph{Witness Oracle,} which can be
efficiently implemented on graph classes forbidding certain structures.

\paragraph*{Forbidden Structures (Obstructions)}

We define here various obstructions that are of interest. Let $G=(L,R,E)$ be a
bipartite graph. The sequence $(a_1,b_1),\dots,(a_n,b_n)\allowbreak\in L\times
    R$ forms
\begin{itemize}
    \item a \emph{co-matching} of order $n$ in $G$ if we have $(a_i,b_j)\in
              E\iff i\neq j$, for all $i,j\in [n]$;
    \item a \emph{ladder} of order $n$ in $G$ if we have $(a_i,b_j)\in E\iff i>
              j$, for all $i,j\in [n]$;
    \item a \emph{semi-ladder} of order $n$ in $G$ if we have $(a_i,b_j)\in E$
          for all $i,j\in [n]$ with $i>j$, and $(a_i,b_i)\notin E$ for all
          $i\in[n]$,
\end{itemize}
and the sequence $(a_1,b_1,c_1),\dots,(a_n,b_n,c_n)\in L\times R\times L$ forms
\begin{itemize}
    \item a \emph{double-ladder} of order $n$ if we have $(a_i,b_j)\in E\iff i>
              j$ and $(b_i,c_j)\in E\iff i>j$, for all $i,j\in [n]$.
\end{itemize}
An illustration of the structures can be found in \cref{fig:obstructions}.

We remark that given a co-matching $(a_1,b_1),\dots,(a_n,b_n)$ of order $n$ and
a subset $X\subseteq [n]$, the subsequence $((a_i,b_i))_{i\in X}$ obtained from
restricting the indices to $X$ forms a co-matching of order $\Abs{X}$. The same
applies to ladders, semi-ladders, and double-ladders. We later make use of this
fact to find ``well-behaved'' obstructions given an arbitrary one.

Since the structures are defined as sequences of tuples of vertices, there can
potentially be repeated occurrences of the same vertex within the sequence. For
a co-matching, a ladder, or a semi-ladder, the neighborhood structures guarantee
that $a_i\neq a_j$ and $b_i\neq b_j$ for $n\geq i>j\geq 1$, as
$b_j\in N(a_i)-N(a_j)$ and $a_i\in N(b_j)-N(b_i)$. For a double-ladder,
the neighborhood
structure guarantees that $a_i\neq a_j$, $b_i\neq b_j$, and $c_i\neq c_j$ for
each $i,j\in [n]$ with $i\neq j$ for the same reason. Moreover, for $i>1$ and
any $j\in [n]$, we have $b_1\in N(a_i)-N(c_j)$. Similarly, for $i\in[n]$ and
$j<n$, we have $b_n\in N(c_j)-N(a_i)$. It follows that $a_i=c_j$ is only
possible when $i=1$ and $j=n$, and indeed $N(a_1)=N(c_n)=\emptyset$ can always
be chosen to be the same vertex in $L$. Therefore, when we refer to these
structures as (sub)graphs, we understand a co-matching, ladder, or semi-ladder
of order $n$ as a bipartite graph with $2n$ vertices, and a double-ladder of
order $n$ as a bipartite graph with $3n-1$ vertices, where $a_1$ and $c_n$ are
realized by the same vertex (see $a_1$ and $c_4$ in \cref{fig:double-ladder}).

Let $G$ be a graph and $H$ be a bipartite graph. We say that $H$ is a
\emph{semi-induced subgraph} of $G$ if there exist disjoint vertex sets
$A,B\subseteq V(G)$ such that the bipartite graph
$(A,B,\Set{uv\in E(G)\given u\in A, v\in B})$
is isomorphic to $H$. The \emph{co-matching index} of a graph is the maximum
order of a co-matching that it contains as a semi-induced subgraph, and the
co-matching index of a class of graphs is the supremum of the co-matching
indices of its members. We also write the \emph{$\dom^k$-co-matching index}
(resp. $\dom$-co-matching index) of a graph $G$ or a class $\C$ to denote the
corresponding index of the agreement graph $\dom^k(G)$ (resp. $\dom(G)$) or
class of agreement graphs $\dom^k(\C)$ (resp. $\dom(\C)$). A graph class $\C$ is
\emph{co-matching-free} if for some $\ell > 0$ and for every graph $G\in \C$,
$G$ does not contain a co-matching of order $\ell$ as a semi-induced subgraph.
We define the same for semi-ladder, ladder, and double-ladder analogously.

\section{Dominating Set Algorithm for the Helly Property}
\label{sec:co-matching-algo}
We introduce the \emph{Co-Matching Algorithm} that solves the
\textsc{Dominating Set} problem,
and show that its running time is fpt if the
$\dom^k$-co-matching index of the graph is bounded, i.e.,
its agreement graph does not contain arbitrarily long co-matchings.
The central, yet simple algorithmic technique of this algorithm is the
so-called ``Compression Step'' that
augments the Semi-Ladder Algorithm of \cite{FabianskiPST19}.
We first describe the algorithm on an abstract level, relying on the existence
of three oracles and analyze the number and input size of the oracle calls in
terms of the co-matching and semi-ladder indices of the agreement graph.
Then, we present an efficient implementation of the required
oracles for the \textsc{Dominating Set} problem and analyze the overall
running time, finalizing the description of the Co-Matching Algorithm.

\subsection{The Algorithm in the Abstract, the Oracles and Redundant Witnesses}
The {Co-Matching Algorithm} relies on access to the \emph{Candidate,}
\emph{Witness,} and \emph{Compression Oracles,} as defined in the following,
to explore the agreement graph without constructing it in full. We
defer their implementation to the next subsection.

\begin{description}
    \item[Candidate Oracle:] Given a set of witnesses $W \subseteq W^*$, the
          oracle either returns a candidate $\vc \in C^*$ that agrees with all
          witnesses in $W$, or it concludes that no such candidate exists.
    \item[Witness Oracle:] Given a candidate $\vc \in C^*$, the oracle either
          returns a witness $w \in W^*$ that does not agree with $\vc$, or it
          concludes that no such witness exists.
    \item[Compression Oracle:] Given a set of witnesses $W \subseteq W^*$, the
          oracle either returns a \emph{redundant} witness (see below) for $W$,
          or it concludes that no such witness   exists.
\end{description}

\begin{definition}[Redundant witness]\label{def:redundant-witness}
    In an agreement graph $H =
        (C^*,W^*,E)$, we say a witness $w$ is \emph{redundant} for a set of
    witnesses $W\subseteq W^*$ if every candidate of $C^*$ that agrees with all
    witnesses of $W-w$ also agrees with $w$.
\end{definition}
In other words, a redundant witness is a witness that can be removed from the
set of witnesses $W$ without enlarging the set of candidates that agree with
$W$. We show that irredundant witness sets are co-matchings in disguise.

\begin{corollary} \label{cor:irredundant} An agreement graph $H = (C^*,W^*,E)$
    has co-matching index at most $\mu$ if and only if every set of witnesses
    $W \subseteq W^*$ of size at least $\mu+1$ contains a redundant witness.
\end{corollary}
\begin{proof}
    Let $H = (C^*,W^*,E)$ be a bipartite graph and $W =
        \Set{w^1,\dots,w^{\mu+1}}\subseteq W^*$ be a set of irredundant
    witnesses of order $\mu+1$. For any $w \in W$, since it is irredundant,
    there exists a candidate $\vc^w \in C^*$
    that agrees with all witnesses in
    $W - w$ but disagrees with $w$. Hence, $(\vc^{w^1},w^1), \dots,
        (\vc^{w^{\mu+1}},w^{\mu+1})$ forms a co-matching of order $\mu+1$ in
    $H$.

    Conversely, if $(C,W)$ forms a co-matching of order $\mu+1$ in $H$, then
    all witnesses in $W$ are irredundant by the same argument.
\end{proof}

The Co-Matching Algorithm extends the Semi-Ladder Algorithm from
\cite{FabianskiPST19} by introducing the Compression Step and Compression
Oracle. Otherwise, these algorithms are identical.

\begin{GrayBox}{\textbf{Co-Matching Algorithm}}
    The algorithm maintains a set of witnesses $W \subseteq W^*$ which is set to
    be empty initially. It repeatedly performs the following three steps until
    either a solution is found or a set of witnesses $W$ is found such that no
    candidate agrees with all witnesses in $W$.
    \begin{itemize}
        \item \textbf{Candidate Step:} The Candidate Oracle is called to find a
              candidate $\vc \in C^*$ that agrees with the current set of
              witnesses $W$. If no such candidate exists, the algorithm returns
              \textsc{no}.
        \item \textbf{Witness Step:} The Witness Oracle is called to find a
              witness $w \in W^*$ that does not agree with $\vc$. If there is no
              such witness, the algorithm returns $\vc$ as the solution.
              Otherwise, add $w$ to~$W$.
        \item \textbf{Compression Step:} The Compression Oracle is called on the
              current set of witnesses~$W$. If a redundant witness $w$ is
              returned, remove $w$ from $W$ and repeat this step. Otherwise,
              proceed with the next round.
    \end{itemize}
\end{GrayBox}

We will argue the correctness and running time of the algorithm first in terms
of oracle calls, and later discuss how to implement the oracles efficiently.
First, we observe that the algorithm constructs a semi-ladder in $\dom^k(G)$ and
maintains a co-matching in $\dom^k(G)$ throughout its run.

\begin{lemma}\label{lem:co-matching-algo-analysis}
    The Co-Matching Algorithm satisfies the following properties when
    applied to an agreement graph $H = (C^*, W^*, E)$ with co-matching index at most
    $\mu$ and semi-ladder index $\lambda$.
    \begin{enumerate}
        \item \textnormal{\textbf{(Correctness)}}
              If the algorithm returns a candidate $\vc$,
              then $\vc$ agrees with all
              witnesses in $W^*$. If the algorithm returns
              \textnormal{\textsc{no}}, then no candidate in $C^*$
              agrees with all witnesses in $W^*$.
        \item \textnormal{\textbf{(Oracle-Call Bounds)}}
              The algorithm makes at most $\lambda+1$ calls to the
              Candidate Oracle, at most
              $\lambda+1$ calls to the Witness Oracle, and at most $2\lambda$
              calls to the
              Compression Oracle.
        \item \textnormal{\textbf{(Witness-Set Bound)}}
              Throughout the execution, $\Abs{W}\leq\mu+1$.
    \end{enumerate}
\end{lemma}
\begin{proof}
    \textbf{(Correctness.)}
    If $\vc$ is returned,
    the Witness Oracle was called with $\vc$ and returned no witness,
    certifying that no witness in $W^*$ disagrees with $\vc$. Hence, $\vc$
    agrees with all of $W^*$.
    If \textsc{no} is returned, the Candidate Oracle was called with some set
    $W\subseteq W^*$ and returned no candidate,
    certifying that no candidate agrees
    with all of $W$. Since $W\subseteq W^*$, any candidate agreeing with all of
    $W^*$ would in particular agree with all of $W$, contradicting the oracle's
    response. Hence, no solution exists.

    \textbf{(Oracle-Call Bounds.)}
    We first bound the number of rounds.
    Let $\vc^1,\dots,\vc^r$ be the candidates returned by the
    Candidate Oracle called on witness sets
    $W^0,\dots,W^{r-1}$,
    and let the subsequent Witness Oracle call return the witnesses
    $w^1,\dots,w^r$.
    By the Candidate Step's contract,
    $\vc^i$ agrees with all witnesses in $W^i$,
    which,
    by the contract of all previous compression oracle calls,
    implies that $\vc^i$ agrees with all witnesses in $\Set{w^1,\dots,w^{i-1}}$.
    By the Witness Step's contract, $w^i$
    disagrees with $\vc^i$.
    Therefore,
    $(\vc^1,w^1),\dots,(\vc^r,w^r)$ forms a semi-ladder of order $r$ in $H$;
    consequently,
    $r\leq\lambda$.

    There are exactly $r$ rounds that produce a witness, plus at most one final
    round where either the Candidate Oracle returns \textsc{no} or the Witness
    Oracle returns \textsc{no}, for a total of at most $\lambda+1$
    calls to each.
    In each of the $r$ witness-producing rounds,
    the Compression Oracle is called at least once.
    Each such call either removes a redundant witness (a successful call)
    or does not (an unsuccessful call). Since a witness can be removed at most
    once, there are at most $r\leq\lambda$ successful calls.
    There is at most one
    unsuccessful call per round, giving at most
    $r\leq\lambda$ unsuccessful calls.
    Note that no compression is done in the final round.
    Hence, the Compression Oracle is called
    at most $2\lambda$ times.

    \textbf{(Witness-Set Bound.)}
    By \cref{cor:irredundant}, after each Compression Step the set $W$ contains
    only irredundant witnesses, so $\Abs{W}\leq\mu$.
    The subsequent Witness Step
    adds exactly one witness, so $\Abs{W}\leq\mu+1$ before the next Compression
    Step. The bound holds throughout the execution.
\end{proof}

We then state the running time of the Co-Matching Algorithm while treating the
oracles as black boxes.

\begin{lemma}\label{thm:co-matching-algo-runtime-abstract}
    Let $H = (C^*, W^*, E)$ be an agreement graph with co-matching index $\mu$
    and semi-ladder index $\lambda$.
    Let $T_{\mathrm{cand}}(m)$, $T_{\mathrm{witn}}$,
    $T_{\mathrm{comp}}(m)$ be the worst-case running times of a single call to
    the Candidate, Witness, and Compression Oracles on inputs of size at most
    $m$.
    The Co-Matching Algorithm runs in time
    \[
        O\Paren[\big]{\lambda(T_{\mathrm{cand}}(\mu+1)+T_{\mathrm{witn}}+
        T_{\mathrm{comp}}(\mu+1))}\,.
    \]
\end{lemma}
\begin{proof}
    By \cref{lem:co-matching-algo-analysis}, there are at most $\lambda+1$
    calls to
    the Candidate and Witness Oracles, each with $\Abs{W}\leq\mu+1$, and at most
    $2\lambda$ calls to the Compression Oracle, also with $\Abs{W}\leq\mu+1$.
    Multiplying call counts by per-call costs yields the stated bound.
\end{proof}

It remains to show that the
oracles can be implemented efficiently for the agreement graph $\dom^k(G)$ to achieve an fpt algorithm
for \textsc{Dominating Set} on graph classes of bounded $\dom^k$-co-matching index.

\subsection{Implementation of the Oracles for Dominating Set}
\label{sec:domset-oracles}

We give an efficient implementation of the oracles for the agreement graph
$\dom^k(G)$ of a graph $G$.
In the following, only the graph $G$ is given as input,
and the agreement graph $\dom^k(G)$ is not constructed explicitly.
As the $\dom^k$-semi-ladder index $\lambda$ may be as large as $O(\Abs{G})$,
this implementation will be a bit more technical than the one for the
Semi-Ladder Algorithm
to keep the polynomial dependency on $\Abs{G}$ as low as possible.

\begin{definition}
    For a set $S$ of vertices in $G$, we call $P \subseteq S$ a
    \emph{realized neighborhood} on $S$ if there exists a
    vertex $v$ with $P = N_G[v] \cap S$.
    The set of realized neighborhoods on $S$ is then
    \[
        \mathcal P^G(S) \coloneq \Set{N_G[v] \cap S\mid v\in V(G)}\subseteq 2^{S}\,.
    \]
\end{definition}
We represent any realized neighborhood by a bitvector of length $\Abs{S}$.
We implement the oracles by the following procedures.

\begin{description}
    \item[Witness Oracle:]
          Given $\vc=(c_1,\dots,c_k)$, compute $\bigcup_{i=1}^k N[c_i]$.
          If this union is $V(G)$, return \textsc{no}; otherwise return any vertex
          not in the union.
    \item[Compression Oracle:]
          We first observe that a witness $w \in W$ is not redundant for $W$ in $\dom^k(G)$
          if and only if there exists a candidate $\vc \in V(G)^k$ that agrees with
          all witnesses in $W - w$ but disagrees with $w$ in $\dom^k(G)$.
          That is, in the original graph $G$, there exists a vertex set
          $S \subseteq V(G)$ of size at most $k$ corresponding to $\vc$ that
          does not dominate $w$ and every vertex in $W - w$ is dominated by some
          vertex in $S$.

          We check the existence of $S$ as follows: First, we compute the set of
          realized neighborhoods $\mathcal P^G(W)$.
          For each possible $k$-combination of realized neighborhoods $P_1,\dots,P_k \in \mathcal P^G(W)$, we check whether
          the union of the realized neighborhoods contains $W\setminus\Set{w}$ but not $w$.
          If such a combination exists,
          we know there exist vertices $v_1,\dots,v_k \in V(G)$ realizing the neighborhoods
          $P_1,\dots,P_k$ respectively, and forming the desired set $S$.
          In such a case, we return that $w$ is not redundant.
          If no such combination exists, we conclude $w$ is redundant.
    \item[Candidate Oracle:]
          Enumerate all $k$-combinations of realized neighborhoods
          whose union is $W$.
          For the first such combination $P_{1},\dots,P_{k}$, materialize a concrete
          tuple $\vc$ by scanning $V(G)$ until one vertex matching each chosen realized neighborhood is found.
          Return $\vc$ in this case or \textsc{no} if no such combination of realized neighborhoods exists.
\end{description}

Before we analyze the running time of the oracle implementation,
we give more details on how to compute $\mathcal P^G(W)$ efficiently.
The implementation maintains an auxiliary bipartite graph
\[
    B(W)\coloneq\Paren[\big]{V(G),\,W,\;\Set*{vw\given v\in V(G),\,w\in W,\;vw\in E(G)\text{ or }v=w}}
\]
that captures, for each vertex $v\in V(G)$, which witnesses are in its closed
neighborhood. Its size is $|W|\cdot |G|$.
We have that the set of realized neighborhoods of $W$ in $G$ equals the set of realized neighborhoods of $W$ in $B(W)$,
i.e. \(
\mathcal P^G(W) =  \mathcal P^{B(W)}(W)
\).
During the run of the Co-Matching Algorithm, we keep $B(W)$ up to date incrementally.
Initially, $W=\emptyset$ and
$B(W)$ has no right-side vertices.  When a witness $w$ is added to $W$,
we scan $N_G[w]$ and add the vertex $w$ together with edges
$\Set{vw\given v\in N_G[w]}$ to $B(W)$.  When $w$ is removed, we delete $w$
and its incident edges.
The implementation of the oracles described above then is executed on $B(W)$ instead of $G$, especially the computation of $\mathcal P^G(W)$ and finding vertices $\vc$ that materialize a specific $k$-combination of realized neighborhoods.

\subsection{Overall Running Time Analysis}
\label{sec:domset-runtime}

We now analyze the per-call running times of the three oracles implemented
as described above, then combine them with the abstract bounds of
\cref{lem:co-matching-algo-analysis,thm:co-matching-algo-runtime-abstract}
to obtain the concrete running time for \textsc{Dominating Set}.
First, we need to introduce
the notion of neighborhood complexity~\cite{ReidlSS19}.

\begin{definition}[Neighborhood complexity]
    The \emph{neighborhood complexity} of a graph $G$ is the function $\nu^G
        \colon \N \to \N$ with
    \[
        \nu^G(m) = \max_{S \subseteq V(G), \Abs{S} \leq m}\Abs*{\mathcal P^G(S)}\,.
    \]
    That is, $\nu^G(m)$ is the maximum number of realized neighborhoods
    on a set of size $m$.
\end{definition}

The neighborhood complexity of a class $\C$ is defined as
$\nu^{\C}(m) = \sup_{G \in\C} \nu^G(m)$.
Since $\mathcal P^G(S)\subseteq 2^S$,
we can always upper bound $\nu^G(m)$ by $2^m$.
For nowhere
dense classes, the  neighborhood complexity is almost linear~\cite{EickmeyerGKKPRS17}
and for $K_{t,t}$-free classes, $\nu^G(m) = O(m^t)$ holds~\cite{FabianskiPST19}.

\begin{lemma}\label{lem:oracle-costs}
    For the agreement graph $\dom^k(G)$, the three oracles can be implemented
    with the following worst-case running times on an input $G$ and a set $W$ of
    witnesses of size $m=\Abs{W}$, assuming the auxiliary graph $B(W)$ is given:
    \[
        \begin{aligned}
            T_{\mathrm{cand}}(m) & = O\Paren[\big]{km\cdot\nu^G(m)^k + m\cdot\Abs{G}}, \\
            T_{\mathrm{witn}}    & = O(k\cdot\Abs{G}),                                 \\
            T_{\mathrm{comp}}(m) & = O\Paren[\big]{km\cdot\nu^G(m)^k}.
        \end{aligned}
    \]
\end{lemma}
\begin{proof}
    The implementation is described in \cref{sec:domset-oracles}. The {Candidate Oracle} enumerates all $k$-combinations of realized neighborhoods whose union
    covers $W$; with at most $\nu^G(m)$ realized neighborhoods, this costs
    $O(km\cdot\nu^G(m)^k)$.  Materializing a concrete tuple from the chosen
    realized neighborhoods requires a scan of $B(W)$ costing $O(m\cdot\Abs{G})$.  The {Witness Oracle} scans the closed
    neighborhoods of the $k$ candidate vertices, costing $O(km\cdot\Abs{G})$.
    The {Compression Oracle} enumerates $k$-combinations of realized neighborhoods for
    each $w\in W$, costing $O(m\cdot\nu^G(m)^k)$.
\end{proof}

\ThmCoMatchingAlgoRuntime
\begin{proof}
    The term $\Norm{G}$ accounts for reading the input initially.
    By \cref{lem:co-matching-algo-analysis}, the algorithm makes at most
    $\lambda+1$ calls to each of the Candidate and Witness Oracles and at most
    $2\lambda+1$ calls to the Compression Oracle, with $\Abs{W}\leq\mu+1$
    throughout.

    Building and maintaining the auxiliary graph $B(W)$ costs $O(\Abs{G})$ per round (one witness added and up to one witness removed on average),
    accumulating to $O(\lambda \cdot \Abs{G})$ overall.
    Each of the at most $\lambda+1$ Witness Oracle calls costs $O(k\cdot\Abs{G})$,
    contributing $O(\lambda k\cdot\Abs{G})$.
    The Candidate Oracle and the Compression Oracle cost
    $O(\mu k\cdot\nu^G(\mu+1)^k + \mu \cdot \Abs{G})$ per call by \cref{lem:oracle-costs};
    There are $O(\lambda)$ calls to each,
    contributing $O(\lambda \mu k\cdot\nu^G(\mu+1)^k + \lambda \mu \Abs{G})$.

    Summing all terms yields
    the stated bound.
\end{proof}

When the input graph class has bounded $\dom^k$-co-matching index, $\mu$ is
bounded by a function of $k$ alone, yielding an overall fpt running time.
\begin{corollary}
    The Co-Matching Algorithm solves the \textnormal{\textsc{Dominating Set}} problem
    in fpt time parameterized by $k$
    on graph classes of bounded $\dom^k$-co-matching indices.
\end{corollary}
A more concrete running time is given in \cref{cor:runtime-class} using the
characterization in \cref{sec:graph-class}.

As the co-matching index of $\dom^k(G)$ is bounded by the semi-ladder index of $\dom^k(G)$, the Co-Matching Algorithm recovers the runtime of the Semi-Ladder Algorithm of \cite[Theorem 21]{FabianskiPST19} on classes of bounded $\dom^k$-semi-ladder index which a slight improvement on the polynomial dependency on $G$.
\begin{corollary} \label{cor:runtime-nwd-biclique}
    The Co-Matching Algorithm solves the \textnormal{\textsc{Dominating Set}} problem
    in time
    $\Norm{G} + k^{O(1)}\Abs{G} + k^{O(k)}$ when $\mathcal C$ is either a nowhere-dense or biclique-free graph class and where $O(\cdot)$ hides constants depending on the graph class $\mathcal C$.
\end{corollary}

\subsection{Connection to Threshold Sets}\label{sec:threshold-set}
In \cite{LokshtanovS21},
the authors identify that \emph{minimal} $k$-domination cores
(see \cref{sec:intro}) are so-called \emph{$k$-threshold sets}.

\begin{definition}[$k$-threshold set~\cite{LokshtanovS21}]
    A vertex set $S$ is a \emph{$k$-threshold set} if for every $v \in S$ there
    exists a set $X_v$ of size at most $k$ so that
    $N [X_v ] \cap S = S\setminus \Set{v}$,
    i.e.,
    $S\setminus\Set{v}$ is a neighborhood realized by $X_v$.
\end{definition}

Using this observation,
they bound the maximum size of $k$-threshold sets
to obtain the fpt result for weakly $\gamma$-closed graphs
(see \cref{def:closed}).

\begin{lemma}[{\cite[Lemmas 20 and 21]{LokshtanovS21}}]
    \label{lem:weakly-closed-threshold-set}
    Every $k$-threshold set of a weakly $\gamma$-closed graph
    has size at most $(\gamma-1)k^{O(\gamma k^2)}$.
\end{lemma}

Using our terminology,
a set of witnesses $W\subseteq W^*$ is a $k$-threshold set if for
every $w \in W$, there exists a candidate $\vc^w$
that agrees with all witnesses
in $W - w$ but disagrees with $w$.
That is, every $w\in W$ is irredundant,
and they
(together with $\Set{\vc^w \given w \in W}$)
form a co-matching of order $\Abs{W}$ in $\dom^k(G)$.
Hence, a graph class has bounded $\dom^k$-co-matching index if and only if
the size of all minimal $k$-domination cores is bounded.

\begin{remark}
    Let $G$ be a graph,
    let $\dom^k(G)=(C^*,W^*,E)$ be its agreement graph,
    and let $W=\Set{w^1,\dots,w^n}\subseteq W^*$.
    Then,
    the following are equivalent:
    \begin{enumerate}
        \item $W$ is a $k$-threshold set in $G$.
        \item Every $w\in W$ is irredundant.
        \item There exists $\vc^1,\dots,\vc^n\in C^*$ such that the sequence
              $(\vc^1,w^1),\dots,(\vc^n,w^n)$ forms a co-matching.
    \end{enumerate}
\end{remark}

The bound on
the size of threshold-sets on weakly $\gamma$-closed graphs~\cite{LokshtanovS21} therefore gives a bound on the
$\dom^k$-co-matching index on those graphs.
We remark that \emph{any} algorithm that ``greedily''
constructs $k$-domination cores to solve or kernelize the
\textsc{Dominating Set} problem cannot work on classes of unbounded
$\dom^k$-co-matching index.

In the next section, we characterize the graph classes of bounded
$\dom^k$-co-matching index via forbidden semi-induced subgraphs.
As an application, we reprove that weakly
$\gamma$-closed graphs have bounded $\dom^k$-co-matching index (i.e.,
bounded-sized minimal $k$-domination cores).
Our bound (\cref{thm:weakly-closed-dom-k-co-matching})
asymptotically matches that of~\cref{lem:weakly-closed-threshold-set}
with improved constants and using simpler proofs.
The resulting running time of the Co-Matching
Algorithm also matches that of the algorithm proposed in~\cite{LokshtanovS21}.

\section{Graph Classes of Bounded
\texorpdfstring{$\dom^k$}{dom\^{}k}-Co-Matching Index}\label{sec:graph-class}
In this section,
we characterize the graph classes that have bounded $\dom^k$-co-matching index,
and relate them to previously studied notions.

\subsection{Characterization of Graphs of Bounded
    \texorpdfstring{$\dom^k$}{dom\^k}-Co-Matching Index}
\label{sec:graph-class-characterization}

We start by showing that boundedness of $\dom^k$-co-matching index is exactly
characterized by forbidding co-matchings and double-ladders as semi-induced
subgraphs.

\ThmCharacterizeCoMatchingIndex

To prove~\cref{thm:characterize-co-matching-index},
we need to be able to argue that long co-matchings in the agreement graph
$\dom^k(G)$ translate to long semi-induced co-matchings or double-ladders
in $G$.
We do this in two steps.

First,
we show that the existence of a long co-matching in $\dom^k(G)$ implies
that of
either a long co-matching or a long double-ladder in $\dom(G)$.
For $k=1$,
$\dom^k(G)$ is exactly $\dom(G)$,
and thus forbidding co-matchings suffice.
For $k\geq 2$,
we make use of Ramsey's theorem,
which states that organized substructures can always be found from large
enough structures.
We state here a multicolor version.
\begin{theorem}[Multicolor Ramsey's theorem]\label{thm:ramsey}
    For any $n\in \N$ and $c\geq 2$,
    there exists a number $R^c(n)\leq c^{cn-1}$
    such that if the edges of a complete graph
    $G$ of order $R^c(n)$ is colored using $c$ colors,
    then there exists a \emph{monochromatic} subset of vertices
    $S\subseteq V(G)$ of size at least $n$,
    i.e.,
    all edges between vertices in $S$ have the same color.
\end{theorem}

In particular,
we color each pair of indices with a pair of colors $(p_1,p_2)$,
where $p_1=p_2$ corresponds to a co-matching and $p_1\neq p_2$ corresponds
to a double-ladder in $\dom(G)$.

\begin{lemma}\label{lem:co-matching-Ramsey}
    Suppose $G$ is a graph whose $\dom$-double-ladder and $\dom$-co-matching
    indices are both smaller than $\ell$. Then, $\dom^k(G)$ has co-matching
    index smaller than $R^{k^2}(\ell)$.
\end{lemma}

\begin{proof}
    Let $G$ be a graph with vertex set $V$ and suppose that $\dom^k(G)$ has
    co-matching index at least $q\coloneq R^{k^2}(\ell)$. Then, there are
    tuples $(\va^1,b^1),\dots,(\va^q,b^q)\in V^{k+1}$ such that for all
    $i,j\in[q]$, $G\models \dom^k(\va^i;b^j)$ if and only if $i\neq j$.

    Since $\dom^k(\vx;y)=\bigvee_{p\in[k]}\dom(x_p;y)$, whenever $G\models
        \dom^k(\va;b)$ and $G\nvDash \dom^k(\va';b')$, there must be some
    $p\in[k]$ such that $G\models \dom(a_p;b)$ and $G\nvDash \dom(a_p';b')$.
    For all $q\geq i>j\geq 1$, we color the pair $(i,j)$ as follows:
    \begin{enumerate}
        \item If for some $p\in [k]$, $G\models \dom(a_p^i;b^j)$ and
              $G\models \dom(a_p^j;b^i)$, then we color $(i,j)$ with
              $(p,p)$.
        \item Otherwise, we color $(i,j)$ with $(p_1,p_2)\in [k]^2$ such
              that $G\models \dom(a_{p_1}^i;b^j)$ and $G\models
                  \dom(a_{p_2}^j;b^i)$. Note that since the previous case does
              not hold for $p_1$ and $p_2$, we have $G\nvDash
                  \dom(a_{p_1}^j;b^i)$ and $G\nvDash \dom(a_{p_2}^i;b^j)$, and
              thus $p_1\neq p_2$.
    \end{enumerate}

    By \cref{thm:ramsey}, there is a set $X\subseteq[q]$ of size $\ell$ and a
    color $(p_1,p_2)\in[k]^2$ such that each pair $(i,j)\in X^2$ with $i>j$
    is colored with the color $(p_1,p_2)$. We distinguish two cases. If
    $p_1=p_2=p$, then
    \begin{itemize}
        \item $G\nvDash \dom(a_{p}^i;b^i)$ for all $i\in X$ and
        \item $G\models \dom(a_{p}^i;b^j)$ for all $i\neq j\in X$.
    \end{itemize}
    It follows that $(a_p^i,b^i)_{i\in X}$ forms a co-matching of order
    $\ell$ in $\dom(G)$.

    On the other hand, if $p_1\neq p_2$, then
    \begin{itemize}
        \item $G\nvDash \dom(a_{p_1}^i;b^i)$ for all $i\in X$,
        \item $G\nvDash \dom(a_{p_2}^i;b^i)$ for all $i\in X$,
        \item $G\nvDash \dom(a_{p_1}^i;b^j)$ for all $i<j\in X$,
        \item $G\nvDash \dom(a_{p_2}^i;b^j)$ for all $i>j\in X$,
        \item $G\models \dom(a_{p_1}^i;b^j)$ for all $i>j\in X$, and
        \item $G\models \dom(a_{p_2}^i;b^j)$ for all $i<j\in X$.
    \end{itemize}
    It follows that $(a_{p_1}^i,b^i,a_{p_2}^i)_{i\in X}$ forms a
    double-ladder of order $\ell$ in $\dom(G)$.
\end{proof}

Next,
we relate the size of co-matchings and double-ladders in $\dom(G)$ to
semi-induced ones.

\begin{lemma}\label{lem:co-matching-semi-induced-delta_1}
    Let $G$ be a graph, let $\ell_1$ denote the co-matching index of
    $\dom(G)$, and let $\ell_2$ denote the order of the largest semi-induced
    co-matching in $G$. Then, $\ell_2\leq \ell_1\leq 2\ell_2$.
\end{lemma}
\begin{proof}
    Since a semi-induced co-matching in $G$ forms a co-matching in
    $\dom(G)$, we have $\ell_2\leq \ell_1$. On the other hand, let
    $(a^1,b^1),\dots,(a^{\ell_1},b^{\ell_1})$ be a co-matching in $\dom(G)$.
    We write $a^i\in B$ to denote that for some $j\in[\ell_1]$, $a^i$ and
    $b^j$ correspond to the same vertex in $G$, and vice versa. We assign
    and relabel the vertices as follows. For each $i\in[\ell_1]$,
    \begin{itemize}
        \item if $a^i\notin B$ and $b^i\notin A$, assign $a^i$ to $U$ and
              $b^i$ to $V$;
        \item if $a^i\notin B$ and $b^i\in A$, assign $a^i$ to $W$ and $b^i$
              to $Y$;
        \item if $a^i\in B$ and $b^i\notin A$, assign $a^i$ to $Y$ and $b^i$
              to $X$;
        \item if $a^i\in B$ and $b^i\in A$, assign both $a^i$ and $b^i$ to
              $Z$.
    \end{itemize}
    In particular,
    note that if $(a^i,b^i)$ is assigned to $(W,Y)$ with $b^i=a^j$ for some
    $j\in[\ell_1]$,
    then $b^j$ cannot appear in $A$,
    justifying the assignment of $(a^j,b^j)$ to $(Y,X)$.
    To see this,
    assume for contradiction that $b^j=a^k$ for some $k\in[\ell_1]$.
    Then,
    we have $a^jb^j\notin E(G)\implies b^ia^k\notin E(G)\implies i=k$,
    which contradicts with the fact that $a^i\notin B$.
    We conclude that $\Set{U,V,W,X,Y,Z}$ forms a partition of the vertices in
    the co-matching in $\dom(G)$.

    We write $U=\Set{u^1,\dots,u^p}$, $V=\Set{v^1,\dots,v^p}$,
    $W=\Set{w^1,\dots,w^q}$, $X=\Set{x^1,\dots,x^q}$,
    $Y=\Set{y^1,\dots,y^q}$, and $Z=\Set{z^1,\dots,z^{2r}}$, such that
    $(u^1,v^1),\allowbreak\dots,\allowbreak(u^p,v^p)$ forms a co-matching,
    $(w^1,y^1),\allowbreak\dots,\allowbreak(w^q,y^q)$ forms a co-matching,
    $(y^1,x^1),\allowbreak\dots,\allowbreak(y^q,x^q)$ forms a co-matching, and
    $(z^1,z^{r+1}),\allowbreak\dots,(z^r,z^{2r})$ forms a co-matching. Note that each
    of these is a semi-induced co-matching in $G$. It follows that $(U\cup
        W\cup \Set{z^1,\dots,z^r},V\cup Y\cup \Set{z^{r+1},\dots,z^{2r}})$ forms
    a semi-induced co-matching of order $p+q+r$. Thus, we have $2\ell_2\geq
        2(p+q+r)\geq p+2q+2r=\ell_1$.
\end{proof}

\begin{lemma}\label{lem:double-ladder-semi-induced-delta_1}
    Let $G$ be a graph, let $\ell_1$ denote the double-ladder index of
    $\dom(G)$, and let $\ell_2$ denote the order of the largest semi-induced
    double-ladder in $G$. Then, $\ell_2\leq \ell_1\leq 3\ell_2$.
\end{lemma}

To prove \cref{lem:double-ladder-semi-induced-delta_1},
we first show that given an arbitrary matching in the double-ladder,
we can find a large subset $X$ of indices that does not intersect with
the matching.
Let $M$ be a matching on a double-ladder $(a^1,b^1,c^1),\dots,(a^n,b^n,c^n)$
and let $X\subseteq [n]$ be an index set.
We write $M[X]\coloneq \Set{a^ib^j\in
        M\given i,j\in X}\cup \Set{b^ic^j\in M\given i,j\in X}$ to denote the
matching induced by restricting $M$ to the indices in $X$.
We have the following lemma.

\begin{lemma}\label{lem:double-ladder-matching} Let
    $(a^1,b^1,c^1),\dots,(a^n,b^n,c^n)$ be a double-ladder, and let $M$ be a
    matching on the double-ladder. Then, there exists an index set
    $X\subseteq [n]$ of size at least $\Ceil{\frac{n}{3}}$ such that
    $M[X]=\emptyset$.
\end{lemma}
\begin{proof}
    We construct $X$ be greedily removing the index in $[n]$ that removes
    the most matched pairs until $M[X]=\emptyset$. Let $p$ denote the number
    of indices whose removal reduces the number of matched pairs by at least
    two, and $q$ denote the number of indices whose removal reduces the
    number of matched pairs by one. For each of the $q$ matched pairs, at
    least one index must remain in $X$, and thus we have $\Abs{X}\geq q$. On
    the other hand, every pair in $M$ must include some $b^i$, and thus
    $n\geq\Abs{M}\geq 2p+q$, i.e., $\Abs{X}=n-p-q\geq p$. It follows that
    \[
        \Abs{X}\geq \Ceil*{\frac{p+q+\Abs{X}}{3}}=\Ceil*{\frac{n}{3}}\,.\qedhere
    \]
\end{proof}

We are now equipped to prove \cref{lem:double-ladder-semi-induced-delta_1}.

\begin{proof}[Proof of \cref{lem:double-ladder-semi-induced-delta_1}]
    Since a semi-induced double-ladder in $G$ forms a double-ladder in
    $\dom(G)$, we have $\ell_2\leq \ell_1$.

    On the other hand, consider a double-ladder
    $(a^1,b^1,c^1),\dots,(a^{\ell_1},b^{\ell_1},c^{\ell_1})$ in $\dom(G)$.
    It is a semi-induced double-ladder in $G$ if no vertex in $V(G)$ appears
    in both $U\cup W$ and $V$. We achieve this by restricting the
    double-ladder to a subset of indices $X\subseteq[\ell_1]$. Recall that
    $a^i\neq a^j$, $b^i\neq b^j$, and $c^i\neq c^j$ for $i\neq j$. Moreover,
    since $G\models \dom(v;v)$ for each $v\in V(G)$, $a^i=b^j$ or $b^i=c^j$
    is possible only if $i>j$. It follows that the vertices that appear
    in both $A\cup C$ and $B$ form a matching $M$ in the double-ladder.

    Therefore, By \cref{lem:double-ladder-matching}, there exists some
    $X\subseteq [\ell_1]$ of size at least $\Ceil{\frac{\ell_1}{3}}$ such
    that when restricting the indices to $X$, each matched pair in $M$
    appears at most once in the double-ladder. It follows that every vertex
    in the resulting double-ladder is distinct, and thus we obtain a
    semi-induced double-ladder of order $\Ceil{\frac{\ell_1}{3}}$.
    Therefore, we have $3\ell_2\geq \ell_1$.
\end{proof}

\Cref{thm:characterize-co-matching-index} then follows from
\cref{lem:co-matching-Ramsey,lem:co-matching-semi-induced-delta_1,lem:double-ladder-semi-induced-delta_1}.

\begin{proof}[Proof of \cref{thm:characterize-co-matching-index}]
    If $\dom^k(\C)$ has unbounded co-matching index, then for each $n\geq
        1$, there exists some $G_n\in\C$ such that $\dom^k(G)$ has a co-matching
    index at least $R^{k^2}(3n)$. By \cref{lem:co-matching-Ramsey},
    $\dom(G)$ has either co-matching index at least $3n$ or double-ladder
    index at least $3n$. In the former case, by
    \cref{lem:co-matching-semi-induced-delta_1}, $G$ contains a co-matching
    of order at least $\frac{3n}{2}\geq n$ as a semi-induced subgraph. In
    the latter case, by \cref{lem:double-ladder-semi-induced-delta_1}, $G$
    contains a double-ladder of order at least $n$ as a semi-induced
    subgraph.

    On the other hand, if for each $n\geq 1$, there exists some $G_n\in \C$
    such that $G_n$ contains a co-matching $(u_1,v_1),\dots,(u_n,v_n)$ as a
    semi-induced subgraph, then, by choosing $\va_i\coloneq
        (u_i,\dots,u_i)\in V^k$ and $b_i\coloneq v_i\in V$, we can construct a
    co-matching of order $n$ in $\dom^k(G_n)$. Similarly, if $G_n$ contains
    a double-ladder $(u_1,v_1,w_1),\dots,(u_n,v_n,w_n)$ as a semi-induced
    subgraph, then, by choosing $\va_i\coloneq (u_i,w_i,w_i,\dots,w_i)\in
        V^k$ and $b_i\coloneq v_i\in V$, we can again construct a co-matching
    of order $n$ in $\dom^k(G_n)$.
\end{proof}

For completeness, we also show that bounded $\dom$-ladder index is
equivalent to forbidding semi-induced ladders.
\begin{lemma}\label{lem:ladder-semi-induced-delta_1}
    Let $G$ be a graph, let $\ell_1$ denote the ladder index of $\dom(G)$,
    and let $\ell_2$ denote the order of the largest semi-induced ladder in
    $G$. Then, $\ell_2\leq \ell_1\leq 2\ell_2$.
\end{lemma}
\begin{proof}
    Since a semi-induced ladder in $G$ forms a ladder in $\dom(G)$, we have
    $\ell_2\leq \ell_1$.

    On the other hand, consider a ladder
    $(a_1,b_1),\dots,(a_{\ell_1},b_{\ell_1})$ in $\dom(G)$. We first
    consider the possible set of vertices that appear in both parts of the
    ladder. Suppose $a_i=b_j$ and $a_{i'}=b_{j'}$. We obviously have $i>j$
    and $i'>j'$. Moreover, if $a_ia_{i'}\in E(G)$, then we have $i'>j$ and
    $i>j'$, and if $a_ia_{i'}\notin E(G)$, then we have $i'\leq j$ and
    $i\leq j'$. We prove that if any $(i,j),(i'j')\in S\subseteq [\ell_1]^2$
    satisfies the conditions above, then $\Abs{S}\leq
        \Ceil{\frac{\ell_1}{2}}$.

    We proceed by induction on $\ell_1$. Note that $(1,j)$ and $(i,\ell_1)$
    never appears in $S$ for $i,j\in[\ell_1]$. If $\ell_1=0$ or $\ell_1=1$,
    $\Abs{S}=0\leq\Ceil{\frac{\ell_1}{2}}$ holds. For $\ell_1>1$, consider
    the vertices If $(i,1)\notin S$ for each $i\in[\ell_1]$,
    $(\ell_1,j)\notin S$ for each $j\in[\ell_1]$, or $(\ell_1,1)\in S$, then
    $\Abs{S}\leq 1+\Ceil{\frac{\ell_1-2}{2}}=\Ceil{\frac{\ell_1}{2}}$, and
    the statement still holds. Finally, consider the case where
    $(t_1,1),(\ell_1,t_2)\in S$ for some $t_1,t_2\in\Set{2,\dots,\ell_1-1}$.
    Then, we have $1<\ell_1$, and thus we have $t_2<t_1$. For any
    $(i'',j'')\in S$, since $1\leq j''<i''\leq \ell_1$, we must have
    $t_2<i''$ and $t_1>j''$. Thus, $S'=S\setminus
        \Set{(t_1,1),(\ell_1,t_2)}\cup \Set{(\ell_1,1),(t_1,t_2)}$ is also a
    valid set. By the induction hypothesis, $\Abs{S}=1+\Abs{S\setminus
            \Set{(t_1,1),(\ell_1,t_2)}\cup \Set{(t_1,t_2)}}\leq
        1+\Ceil{\frac{\ell_1-2}{2}}=\Ceil{\frac{\ell_1}{2}}$, which completes
    the proof.
\end{proof}

As a corollary, every semi-ladder-free graph class has bounded
$\dom^k$-semi-ladder index. The proof follows from
\cref{lem:ladder-semi-induced-delta_1,lem:co-matching-semi-induced-delta_1}
and \cite[Lemmas 2 and 4]{FabianskiPST19}.
\begin{corollary}\label{cor:semi-ladder-free-delta_1^k} Let $\C$ be a
    semi-ladder-free class. Then, $\dom^k(\C)$ has bounded semi-ladder
    index. In particular, if $\C$ forbids semi-ladders of order $\ell$ as
    semi-induced subgraphs, then $\dom^k(\C)$ has semi-ladder index less
    than $R^{2k}(2\ell)$.
\end{corollary}

Similarly,
every ladder-free graph class has bounded $\dom^k$-ladder index.
The proof follows from \cref{lem:ladder-semi-induced-delta_1} and
the fact that stability is closed under arbitrary Boolean combinations.
\begin{corollary}\label{cor:ladder-free-delta_1^k} Let $\C$ be a
    ladder-free class. Then, $\dom^k(\C)$ has bounded ladder
    index.
\end{corollary}

\paragraph{Hardness of (Double)-Ladder-Free and Co-Matching-Free Classes}
We conclude the characterization with a remark on the hardness of
\textsc{Dominating Set}.
On the agreement graph,
the authors in \cite{FabianskiPST19} show via the
Semi-Ladder Algorithm that having bounded semi-ladder index
suffices for an fpt algorithm for \textsc{Dominating Set},
which is equivalent to having both bounded ladder index and co-matching index
in the agreement graph.
We extend this result by showing that it suffices to have bounded co-matching
index on the agreement graph by introducing the Co-Matching Algorithm.
With the characterization in
\cref{cor:semi-ladder-free-delta_1^k,thm:characterize-co-matching-index},
we see that the Semi-Ladder Algorithm works on classes that are both
ladder-free and co-matching-free,
while the Co-Matching Algorithm relaxes the ladder-free condition
with the weaker notion of double-ladder-free.

These results are tight in the following sense:
As identified in \cite{DreierMS26},
if one drops the condition of
(double-)ladder-freeness or co-matching-freeness,
i.e., only forbids co-matchings
or (double-)ladders as semi-induced subgraphs,
then there exists such a class
where \textsc{Dominating Set} is $\W{1}$-hard.
\begin{proposition}[{\cite[Theorem~5 and Corollary~2]{DreierMS26}}]
    \label{prop:domset-hardness}
    There exists a
    co-matching-free graph class and a ladder-free graph class on which the
    \textnormal{\textsc{Dominating Set}} problem is $\W{1}$-hard.
\end{proposition}

Moreover,
this implies that \textsc{Dominating Set} is hard on classes of bounded $\dom^k$-ladder index using \cref{cor:ladder-free-delta_1^k}.
Hence,
there is no hope to extend the progressive exploration framework to classes of bounded
$\dom^k$-ladder index, as we did for bounded $\dom^k$-co-matching index.

\begin{remark}\label{rem:dom-k-ladder-hardness}
    There exists a class of bounded $\dom^k$-ladder index
    on which \textsc{Dominating Set} is $\W{1}$-hard.
\end{remark}

\subsection{VC-Dimension and Neighborhood Complexity of Co-Matching-Free Classes}

We show now that we can improve the bound on the neighborhood complexity by analyzing the \emph{VC-dimension} of set systems of co-matching-free graph
classes.
A \emph{set system} $(U,\F)$ consists of a universe $U$ and a collection
$\F$ of subsets of $U$. The \emph{Vapnik-Chervonenkis dimension}
(\emph{VC-dimension}) of a set system $(U,\F)$ is the largest subset of $U$
that can be \emph{shattered} by $\F$. That is, the largest $A\subseteq U$
such that for every $A'\subseteq A$, there exists some $S\in \F$ such that
$A'=A\cap S$. The VC-dimension of a graph $G$ is the VC-dimension of the set
system induced by the close neighborhood of its vertices, i.e., $U=V(G)$ and
$\F=\Set{N[v]\given v\in V(G)}$. Note that the incidence graph of such a set
system is isomorphic to $\dom(G)$. We show that the VC-dimension of $G$ can
be bounded by the size of forbidden structures in $\dom(G)$.

\begin{lemma}\label{lem:delta_1-VC-dimension} Let $G$ be a graph and $H$ be
    a bipartite graph with vertices $A\cup B$. If $N(a_1)\neq N(a_2)$ for
    any distinct $a_1,a_2\in A$, and $H$ is not an induced subgraph of
    $\dom(G)$, then the VC-dimension of $G$ is at most $\Abs{B}-1$.
\end{lemma}
\begin{proof}
    Suppose that a set $S\subseteq V(G)$ of size $\Abs{B}$ is shattered. Let
    $f:B\to S$ be an arbitrary bijection, and for each $S'\subseteq S$, let
    $v_{S'}\in V(G)$ be such that $N[v_{S'}]\cap S=S'$. Note that for any
    $a_1\neq a_2$ in $A$, we have $v_{f(N(a_1))}\neq v_{f(N(a_2))}$. Thus,
    by mapping each $a\in A$ to $v_{f(N(a))}$ and each $b\in B$ to $f(b)$,
    we have $f(b)\in N[v_{f(N(a))}]$ if and only if $ab\in E(H)$ for any
    $(a,b)\in A\times B$, i.e., $\dom(G)$ contains $H$ as an induced
    subgraph.
\end{proof}

Using the Sauer-Shelah lemma, we obtain a polynomial bound on the neighborhood
complexity of graphs of bounded $\dom$-co-matching index.
\begin{corollary}\label{cor:delta_1-co-matching-VC-dimension} Every graph
    $G$ with $\dom$-co-matching index at most $\ell$ has VC-dimension at
    most $\ell$, and therefore a neighborhood complexity of
    $\nu^G(m)=O(m^{\ell})$.
\end{corollary}
\begin{proof}
    Let $(u^1,v^1),\dots,(u^n,v^n)$ be a co-matching. Note that
    $N(u^i)-N(u^j)=\Set{v^j}\neq\emptyset$ for $i\neq j$. Therefore, by
    \cref{lem:delta_1-VC-dimension}, the VC-dimension of $G$ is at most
    $\ell$.
\end{proof}
Therefore, every graph class of bounded $\dom$-co-matching index has a
polynomial neighborhood complexity,
which can be used to tighten the running time of
the Co-Matching Algorithm,
as formally stated in \cref{cor:runtime-class}.

\CorRuntimeClass

\begin{proof}
    We use the trivial bound that the $\dom^k$-semi-ladder index $\lambda$
    is at most $\Abs{G}$.
    For the $\dom^k$-co-matching index $\mu$,
    by \cref{lem:co-matching-Ramsey,lem:co-matching-semi-induced-delta_1,lem:double-ladder-semi-induced-delta_1},
    we have $\mu<R^{k^2}(3\ell+1)=k^{O(k^2\ell)}$.
    For the neighborhood complexity,
    by \cref{cor:delta_1-co-matching-VC-dimension},
    we have $\nu^G(m)=O(m^\ell)$,
    and thus $\nu^G(\mu+1)^k=\mu^O(k\ell)=k^{O(k^3\ell^2)}$.
    It follows that by \cref{thm:co-matching-algo-runtime},
    the running time of the Co-Matching Algorithm is
    \[
        O\Paren*{\Norm{G} + \lambda\cdot k\cdot\Abs{G} +
            \lambda\cdot\mu\cdot\Abs{G}+ \lambda\cdot\mu\cdot\nu^G(\mu+1)^k}
        =O\Paren*{k^{O(\ell k^2)}\cdot\Abs{G}^2+
            k^{O(\ell^2 k^3)}\cdot\Abs{G}}\,.\qedhere
    \]
\end{proof}

\subsection{\texorpdfstring{$\dom^k$}{dom\^k}-Co-Matching Index of Weakly \texorpdfstring{$\gamma$}{γ}-Closed Graphs}
We consider here the notion of \emph{weakly $\gamma$-closed} graphs.
The notion of (weak) closure derives from the triadic closure principle in
social networks~\cite{fox20}.

\begin{definition}[$c$-closed and weakly $\gamma$-closed \cite{LokshtanovS21}] \label{def:closed}
    Given an integer $c\geq 1$,
    a graph $G$ is \emph{$c$-closed} if for every $uv \notin E(G)$, the number
    of common neighbors of $u$ and $v$ is less than $c$, i.e.,
    $\Abs{N(u) \cap N(v)} < c$.

    Given an integer $\gamma\geq 1$,
    a graph $G$ is \emph{weakly $\gamma$-closed} if for every
    induced subgraph $H$ of $G$, there exists a vertex $v \in V(H)$ such that
    for every $u \in V(H)$ with $uv \notin E(H)$,
    it holds that $\Abs{N_H(u)\cap N_H(v)} < \gamma$.
    A class $\C$ of graphs has \emph{bounded weak closure}
    if there exists a constant $\gamma$ such that every graph in $\C$
    is weakly $\gamma$-closed.
\end{definition}
While $c$-closed graphs generalize graphs of bounded degree,
weakly $\gamma$-closed graphs generalize graphs of bounded degeneracy.

In \cite{LokshtanovS21}, it is proven that the $\dom^k$-co-matching index,
i.e., the maximum size of all minimal $k$-domination cores
(see \cref{sec:threshold-set}),
of a weakly $\gamma$-closed graph is at most
$((\gamma-1)k+1)R^{3^{15}k^2}(3\gamma)$.
Using the characterization from \cref{thm:characterize-co-matching-index},
we can easily reprove the result and improve the bound to
$R^{k^2}(3\gamma)$.

We start by bounding the $\dom$-co-matching and $\dom$-double-ladder index
of weakly $\gamma$-closed graphs,
as formally stated in
\cref{lem:weakly-closed-co-matching,lem:weakly-closed-double-ladder}.

\begin{lemma}\label{lem:weakly-closed-co-matching}
    Let $G$ be a weakly $\gamma$-closed graph. Then, $\dom(G)$ has co-matching
    index smaller than $\gamma+2$.
\end{lemma}
\begin{proof}
    Let $G$ be a weakly $\gamma$-closed graph,
    and let $\dom(G)=(C^*,W^*,E)$ be its agreement graph (with $k=1$).
    Suppose $C\coloneq\Set{c^1,\dots,c^{\gamma+2}}\subseteq C^*$ and
    $W\coloneq\Set{w^1,\dots,w^{\gamma+2}}\subseteq W^*$ be such that
    the sequence $(c^1,w^1),\dots,(c^{\gamma+2},w^{\gamma+2})$ forms a
    co-matching of order $\gamma+2$ in $\dom(G)$.
    That is, $c^iw^j\in E$ if and only if $i\neq j$.
    Consider the induced subgraph $G'\coloneq G[C\cup W]$.
    We show for each vertex $u$ in $G'$,
    there is a non-adjacent vertex $v$
    that shares at least $\gamma$ neighbors with it.
    In particular,
    we identify a pair of vertices in $\dom(G)$ together with
    a set of at least $\gamma$ common neighbors $S\subseteq N(u)\cap N(v)$.
    Although the same vertex in $G$ may in general
    appear in both $C^*$ and $W^*$,
    since $uv\notin E(G)$,
    they must be distinct from every vertex in $S$.
    It will be clear from our choice of $S$
    that there are no duplicate vertices of $G$ in $S$ either.

    For each $u=c^i$,
    If $C\subseteq N[c^i]$,
    we choose $v\coloneq w^i$ and $S\coloneq A-c^i$.
    We thus have $c^iw^i\notin E$,
    and $\Abs{N(c^i)\cap N(w^i)}\geq \Abs{A-c^i}=\gamma+1$.
    Otherwise,
    we choose $v\coloneq c^j\in C-N[c^i]$ and $S\coloneq W-w^i-w^j$.
    Then, $\Abs{N(c^i)\cap N(c^j)}\geq\Abs{B-w^i-w^j}=\gamma$.

    The case for $u\in W$ is analogous.
\end{proof}

\begin{lemma}\label{lem:weakly-closed-double-ladder}
    Let $G$ be a weakly $\gamma$-closed graph. Then, $\dom(G)$ has
    double-ladder index smaller than $3\gamma$.
\end{lemma}
\begin{proof}
    Let $G$ be a weakly $\gamma$-closed graph,
    and let $\dom(G)=(C^*,W^*,E)$ be its agreement graph.

    Suppose $A\coloneq\Set{a^1,\dots,a^{3\gamma}}\subseteq C^*$,
    $B\coloneq\Set{b^1,\dots,b^{3\gamma}}\subseteq W^*$,
    and $C\coloneq\Set{c^1,\dots,c^{3\gamma}}\subseteq C^*$
    be such that
    $(a^1,b^1,c^1),\dots,(a^{3\gamma},b^{3\gamma},c^{3\gamma})$ forms a
    double-ladder in $\dom(G)$ of order $q\coloneq 3\gamma$.
    That is,
    $a^ib^j\in E$ if and only if $i>j$ and $b^ic^j\in E$ if and only if
    $i>j$.
    We write $A^{i,j}\coloneq \Set{a^t\given i\leq t\leq j}$ and
    similarly $B^{i,j}$ and $C^{i,j}$ to ease notation.
    Consider the induced subgraph
    $G'=G[A^{\gamma+1,3\gamma}\cup B\cup C^{1,2\gamma}]$.
    We show that for each vertex $u$ in $G'$,
    there is a non-adjacent vertex $v$ that shares at least
    $\gamma$ neighbors with it.
    In particular,
    we identify a pair of vertices in $\dom(G)$ together with
    a set of at least $\gamma$ common neighbors $S\subseteq N(u)\cap N(v)$.
    Although the same vertex in $G$ may in general
    appear in both $C^*$ and $W^*$,
    since $uv\notin E(G)$,
    they must be distinct from every vertex in $S$.
    It will be clear from our choice of $S$
    that there are no duplicate vertices of $G$ in $S$ either.

    First consider some $a^i\in A^{\gamma+1,2\gamma}$.
    If $A^{2\gamma+1,3\gamma}\subseteq N(a^i)$,
    then we choose $v\coloneq b^i$ and $S\coloneq A^{2\gamma+1,3\gamma}$.
    We have $a^ib^i\notin E$ and
    $\Abs{N(a^i)\cap N(b^i)}\geq \Abs{A^{2\gamma+1,3\gamma}}=\gamma$.
    Otherwise,
    we choose $v\coloneq a^j\in A^{2\gamma+1,3\gamma}-N(a^i)$
    and $S\coloneq B^{1,\gamma}$.
    We have,
    $\Abs{N(a^i)\cap N(a^j)}\geq \Abs{B^{1,\gamma}}=\gamma$.
    The case for $u\in C^{\gamma+1,2\gamma}$ is analogous:
    simply replace $S=A^{2\gamma+1,3\gamma}$ with $C^{1,\gamma}$
    and $S=B^{1,\gamma}$ with $B^{2\gamma+1,3\gamma}$.

    Next, suppose $u=a^i\in A^{2\gamma+1,3\gamma}$.
    If $C^{1,\gamma}\subseteq N(a^i)$,
    we choose $v\coloneq b^i$ and $S\coloneq C^{1,\gamma}$.
    We have $a^ib^i\notin E$ and
    $\Abs{N(a^i)\cap N(b^i)}\geq \Abs{C^{1,\gamma}}=\gamma$.
    Otherwise,
    we choose $v\coloneq c^j\in C^{1,\gamma}-N(a^i)$ and
    $S\coloneq B^{\gamma+1,2\gamma}$.
    We then have
    $\Abs{N(a^i)\cap N(c^j)}\geq\Abs{B^{\gamma+1,2\gamma}}=\gamma$.
    The case for $u\in C^{1,\gamma}$ is analogous:
    simply replace $S=C^{1,\gamma}$ with $A^{2\gamma+1,3\gamma}$.

    Finally, consider the case where $u=b^i\in B^{1,\gamma}$.
    If $B^{\gamma+1,2\gamma}\subseteq N(b^i)$,
    we choose $v\coloneq c^i$ and $S\coloneq B^{\gamma+1,2\gamma}$.
    We have $b^ic^i\notin E$ and
    $\Abs{N(b^i)\cap N(c^i)}\geq \Abs{B^{\gamma+1,2\gamma}}=\gamma$.
    Otherwise,
    we choose $v\coloneq b^j\in B^{\gamma+1,2\gamma}-N(b^i)$
    and $S\coloneq A^{2\gamma+1,3\gamma}$.
    We have
    $\Abs{N(b^i)\cap N(b^j)}\geq \Abs{A^{2\gamma+1,3\gamma}}=\gamma$.
    The case for $u\in B^{\gamma+1,3\gamma}$ is analogous:
    for $u\in B^{2\gamma+1,3\gamma}$,
    replace $S=A^{2\gamma+1,3\gamma}$ with $C^{1,\gamma}$,
    and for $u\in B^{\gamma+1,2\gamma}$,
    replace $S=B^{\gamma+1,2\gamma}$ with $B^{1,\gamma}$.
\end{proof}

Combining
\cref{lem:co-matching-Ramsey,lem:weakly-closed-co-matching,lem:weakly-closed-double-ladder},
we can bound the $\dom^k$-co-matching index of weakly $\gamma$-closed
graphs.

\ThmWeaklyClosedDomkCoMathing

\begin{proof}
    By \cref{lem:weakly-closed-co-matching},
    $\dom(G)$ has co-matching index smaller than $\gamma+2\leq3\gamma$.
    Similarly, by
    \cref{lem:weakly-closed-double-ladder},
    $G$ has double-ladder index smaller than $3\gamma$.
    By \cref{lem:co-matching-Ramsey}, $\dom^k(G)$ has co-matching index
    smaller than $R^{k^2}(3\gamma)$.
\end{proof}

By \cref{cor:runtime-class},
the running time of the Co-Matching Algorithm is
$k^{O(\gamma^2 k^3)}\Abs{G}^{O(1)}$,
which matches that of the algorithm given in \cite{LokshtanovS21}.

\CorWeaklyClosedRuntime

Finally,
combining
\cref{lem:weakly-closed-co-matching,cor:delta_1-co-matching-VC-dimension},
we can improve the bound on VC-dimension of weakly $\gamma$-closed graphs
from $6\gamma$~\cite[Theorem~3]{LokshtanovS21} to $\gamma+1$,
almost matching the lower bound $\gamma$ given in \cite{LokshtanovS21}.

\subsection{Relation to Different Graph Classes}
We investigate here the properties of graph classes of bounded
$\dom^k$-co-matching index and how they relate to other well-studied graph
classes.
Since all biclique-free classes have bounded semi-ladder index,
they also have bounded $\dom^k$-co-matching index.
We consider dense classes and show two non-inclusion results.
On the one hand,
we show that classes of bounded shrubdepth do not have bounded
$\dom^k$-co-matching index.
On the other hand,
we show that $c$-closed graphs (\cref{def:closed}),
which have bounded $\dom^k$-co-matching index,
have unbounded radius-$1$ merge-width.
These two results suggest that the notion of bounded $\dom^k$-co-matching index
is incomparable to well established tractability frontiers in the dense regime.

We start by showing that the classes of graphs of bounded \emph{shrubdepth,}
one of the simplest known hereditary dense graph classes and a generalization
of bounded treedepth, have unbounded co-matching index.

A functionally equivalent notion is the \emph{SC-depth,} which we define as
follows. Given a graph $G$ and a set of vertices $A\subseteq V(G)$, applying
a \emph{flip} on $A$ in $G$ results in the graph $G\oplus A$ obtained from
$G$ by inverting the adjacency relation within $A$. The \emph{SC-depth} of
$G$ is then defined as
\[
    \SCd(G)=\begin{dcases*}
        0                                                        & if $G$ has one vertex;      \\
        \max\Set{\SCd(C)\given \text{$C$ is a component of $G$}} & if $G$ is disconnected; and \\
        1+\min_{A\subseteq V(G)}\SCd(G\oplus A)                  & otherwise.
    \end{dcases*}
\]
For classes of shrubdepth at most $d$, the \textsc{Dominating Set} problem
can be solved in time $O(2^{dk} n^{O(1)})$ using at most $n^{O(1)}$
space~\cite{BergougnouxCGKMOPL23}.

We show here that the class of co-matchings has bounded shrubdepth, and
thus graph classes of bounded shrubdepth do not have bounded
$\dom^k$-co-matching index.

\ThmShrubdepthCoMatching
\begin{proof}
    We show that the
    SC-depth of co-matchings is at most $4$. Let $G=(U\cup V,E)$, where
    $U=\Set{u_1,\dots,u_n}$ and $V=\Set{v_1,\dots,v_n}$ are such that
    $(u_1,v_1),\dots,(u_n,v_n)$ forms a co-matching. The graph $G'\coloneq
        G\oplus (U\cup V)\oplus U\oplus V$ is exactly a perfect matching, i.e.,
    each connected component in $G'$ is an isolated edge. The SC-depth of
    $G'$ is $1$, and thus the SC-depth of $G$ is at most $4$.
\end{proof}

We next investigate how they
relate to other graph classes. Since graph classes of bounded $\dom^k$-co-matching index generalize graph classes of bounded
degeneracy, they are not monadically dependent.
\emph{Merge-width}~\cite{DreierT25} is a family of graph parameters that
falls outside the monadically stable regime and generalizes the parameters clique-width, twin-width and (weak) coloring numbers. In particular, graph classes
of bounded radius-$1$ merge-width generalize those of bounded degeneracy in
the dense setting, and the two coincide when the graph class is
biclique-free.

\begin{theorem}[{\cite[Corollary~7.7]{DreierT25}}]\label{thm:mw_sparse} A
    graph class $\C$ has bounded degeneracy if and only if
    $\mw_1(\C)<\infty$, and it excludes some biclique $K_{t,t}$ as subgraph.
\end{theorem}

We show that weakly $\gamma$-closed graph classes do not coincide with
classes of bounded degeneracy when projected to biclique-freeness, and
therefore have unbounded radius-$1$ merge-width.

\ThmWeaklyClosedMergeWidth

We prove this by constructing a class of
graphs that is $8$-closed, $K_{3,13}$-free, and has unbounded degeneracy.
The remainder of the section is devoted to this construction.

For each prime $p$, let $P_p=\Set{0,\dots,p-1}$ to ease notation. We
construct a graph $G_p$ with vertex set $V_p\coloneq P_p^2$. We use
$a,b,c,d,x,y$ for elements in $P_p$, and $u,v,w$ for elements in $V_p$.

\begin{definition}
    We call a set of vertices $S\subseteq V_p$ a \emph{line} if
    $S=\Set{(x,y)\given ax+by\equiv c \bmod p}$ for some $a,b,c\in P_p$ with
    $(a,b)\neq (0,0)$.
\end{definition}

Since $\Z_p$ is a field, every line contains exactly $p$ vertices, and any
two lines intersect at more than one point if and only if they are the same
line. Moreover, two lines $S_1=\Set{(x,y)\given a_1x+b_1y\equiv c_1 \bmod
        p}$ and $S_2=\Set{(x,y)\given a_2x+b_2y\equiv c_2 \bmod p}$ are the same if
and only if for some $x\in P_p$, $a_1 x\equiv a_2\bmod p$, $b_1 x\equiv
    b_2\bmod p$, and $c_1x \equiv c_2\bmod p$.

Let $L_{(a,b)}$ denote the line through $(0,a)$ and $(1,b)$ modulo $p$,
i.e., $L_{(a,b)}\coloneq\Set{(x,y)\in V_p\given y\equiv (b-a)x+a\bmod{p}}$.
We show that every vertex in $V_p$ corresponds to a unique line.
\begin{lemma}\label{lem:L-unique} For $(a,b)\neq (c,d)$, $L_{(a,b)}\neq
        L_{(c,d)}$.
\end{lemma}
\begin{proof}
    If $L_{(a,b)}=L_{(c,d)}$, then there exists $x\in P$ such that $x\equiv
        1 \bmod p$, $(b-a)x\equiv (d-c)\bmod p$, and $ax\equiv d\bmod p$. It
    follows that $(a,b)=(c,d)$.
\end{proof}

We assign $L_{x,y}$ to be the set of neighbors of $(x,y)$, and iteratively
add edges until the resulting graph fulfills our requirements.

We first address the symmetry issue. Let $T_{(a,b)}=\Set{(x,y)\given
        (a,b)\in L_{x,y}}$. If we take $L_{(a,b)}\cup T_{(a,b)}$ to be the
neighborhood of $(a,b)$, then, by construction, $(c,d)\in L_{(a,b)}$ if and
only if $(a,b)\in T_{(c,d)}$, and the symmetry of the adjacency relation is
guaranteed.

It remains to show that the $T_{(a,b)}$'s are also lines.

\begin{lemma}\label{lem:T-unique} For any $(a,b)\in P^2$, $T_{(a,b)}$ is a
    line. Moreover, for $(a,b)\neq (c,d)$, $T_{(a,b)}\neq T_{(c,d)}$.
\end{lemma}
\begin{proof}
    By definition,
    \[
        \begin{array}{RCL}
            T_{(a,b)} & = & \Set{(x,y)\given (a,b)\in L_{x,y}}             \\
                      & = & \Set{(x,y)\given b\equiv (y-x)a+x\bmod{p}}     \\
                      & = & \Set{(x,y)\given ay\equiv (a-1)x+b\bmod{p}}\,,
        \end{array}
    \]
    which is a line.

    Assume that $T_{(a,b)}=T_{(c,d)}$. Then, for some $x\in P$, we have
    $ax\equiv c\bmod p$, $(a-1)x\equiv c-1\bmod p$, and $bx\equiv d\bmod p$.
    From the first two conditions, we get $x=1$, and thus $a=c$ and $b=d$.
\end{proof}

To keep the closure of the graph low, we add edges between vertices that
share too many neighbors. Note that for any $u,v\in V_p$ with $u\neq v$,
\[
    \begin{array}{RCL}
        \Abs*{(L_u\cup T_u)\cap (L_v\cup T_v)} & \leq & \Abs*{L_u\cap L_v}+\Abs{L_u\cap T_v}+\Abs{T_u\cap L_v}+\Abs{T_u\cap T_v} \\
                                               & \leq & \Abs{L_u\cap T_v}+\Abs{T_u\cap L_v}+2\,.
    \end{array}
\]
It suffices to account for the cases where $L_u=T_v$ or $T_u=L_v$. By
\cref{lem:L-unique,lem:T-unique}, we can define $f(u)$ to be the unique
vertex such that $L_u=T_{f(u)}$ if it exists, and similarly $g(u)$ the
unique vertex such that $T_u=L_{g(u)}$ if it exists. If such a vertex does
not exist, we define $f(u)=u$ and $g(u)=u$, so that $f$ and $g$ are
well-defined on $V_p$ and $g(f(u))=f(g(u))=u$ holds for all $u\in V_p$.

The graph $G_p=(V_p,E_p)$ is defined such that the open neighborhood of $u$
is $N(u)=L_u\cup T_u\cup\Paren*{\Set{f(u),g(u)}\setminus\Set{u}}$.
We now prove \cref{thm:weakly-closed-merge-width} by showing that $G_p$ has
the claimed properties.

\begin{proof}[Proof of \cref{thm:weakly-closed-merge-width}]
    We show that the graph class $\C=\Set{G_p\given \text{$p$ is prime}}$
    is $8$-closed, $K_{3,13}$-free, and has unbounded
    degeneracy. We first show that $\C$ has unbounded degeneracy by showing
    that $G_p$ has minimum degree at least $p-1$. For every $u\in V_p$,
    \[
        \deg(u)=\Abs{N(u)}\geq\Abs{L_u-u}\geq\Abs{L_u}-1= p-1\,.
    \]

    Next, we show that it is $8$-closed. That is, for each $u\in V_p$ and
    each $v\in V_p-N[u]$, $\Abs{N(u)\cap N(v)}\leq 8$. Note that
    \[
        \begin{array}{RCL}
            \Abs{N(u)\cap N(v)} & \leq & \Abs*{(L_u\cup T_u)\cap (L_v\cup T_v)} +\Abs{\Set{f(u),g(u),f(v),g(v)}}    \\
                                & \leq & \Abs*{L_u\cap L_v}+\Abs{L_u\cap T_v}+\Abs{T_u\cap L_v}+\Abs{T_u\cap T_v}+4 \\
                                & \leq & 8\,.
        \end{array}
    \]

    Finally, we show that $G_p$ is $K_{3,13}$-free. For any three distinct
    vertices $u,v,w\in V_p$,
    \[
        \begin{array}{RCL}
            \Abs{N(u)\cap N(v)\cap N(w)} & \leq & \Abs*{\Paren[\Big]{(L_u\cup T_u)\cap (L_v\cup T_v)}\cap(L_w\cup T_w)}      \\
                                         &      & {}+\Abs{\Set{f(u),g(u),f(v),g(v),f(w),g(w)}}                               \\
                                         & \leq & \Abs{L_u\cap L_v}+\Abs{T_u\cap T_v}+ \Abs{L_u\cap T_v\cap(L_w\cup T_w)}    \\
                                         &      & {}+\Abs{T_u\cap L_v\cap(L_w\cup T_w)}+6                                    \\
                                         & \leq & \Abs{L_u\cap L_w}+\Abs{ T_v\cap T_w}+\Abs{T_u\cap T_w}+\Abs{L_v\cap L_w}+8 \\
                                         & \leq & 12\,.
        \end{array}
    \]

    Since $\C$ excludes $K_{13,13}$ as subgraph and has unbounded
    degeneracy, by \cref{thm:mw_sparse}, $\C$ has unbounded radius-$1$
    merge-width.
\end{proof}

\section{Beyond Dominating Set}\label{sec:beyond-domset}
We turn now our attention to applying the Co-Matching Algorithm to more
problems than \textsc{Dominating Set}.
As our work is an extension to \cite{FabianskiPST19},
we discuss first the settings considered in the paper beyond
\textsc{Dominating Set},
namely the distance-$r$ variant and other types of domination problems,
as well as \textsc{Independent Set},
and show how our results apply in these settings.
We further consider the \textsc{Partial Dominating Set} problem and give
a randomized approximation algorithm using the Semi-Ladder Algorithm.
Finally,
we draw inspirations from \cite{Guillemot25} and extend the framework to
set systems.

\subsection{From Distance-\texorpdfstring{$1$}{1} to Distance-\texorpdfstring{$r$}{r} Dominating Set}
\label{sec:dist-r-ds}
We discuss here how our results can be applied to the
\textsc{Distance-$r$ Dominating Set} problem.
We start with some terminology.
For $r\in\N$, we write $\dom_r(x,y)$ to denote the formula
$\dist(x,y)\leq r$.
The \textsc{Distance-$r$ Dominating Set} problem is defined by the formula
\[
    \dom_r^k(\vx; y)  \coloneq \bigvee_{i=1}^k \dom_r(x^i ; y)\,.
\]

For the Co-Matching Algorithm, the adaptation of the oracles is
straightforward. It suffices to extend the notion of neighborhood complexity
to distance-$r$ neighborhood complexity, yielding an fpt algorithm on graph
classes of bounded $\dom_r^k$-co-matching index.

For the characterization,
since $\dom_r(G)$ and $\dom(G^r)$ are isomorphic,
\begin{theorem}
    \label{thm:distance-r-characterization}
    Let $k\in\N$, let $\C$ be a graph class, and let $r\ge 1$.
    Then $\C$ has bounded $\dom_r^k$-co-matching index if and only if
    the class $\Set{G^r\mid G\in\C}$ of $r$-th powers has bounded
    $\dom^k$-co-matching index.
    Consequently, the Co-Matching Algorithm solves
    \textsc{Distance-$r$ Dominating Set} on graph classes whose $r$-th powers
    are co-matching- and double-ladder-free.
\end{theorem}
However, it is not clear which well-studied graph classes have bounded
$\dom_r^k$-co-matching index but unbounded $\dom_r^k$-semi-ladder index. As
far as we are aware, there is not much known about the latter classes, and it would
be interesting to explore this connection further.

\subsection{Domination-Type Problems}
\label{sec:domination-type}
The Semi-Ladder Algorithm in \cite{FabianskiPST19} applies to more general
\emph{domination-type problems},
which is any problem
defined by an FO formula of the form $\exists\vx\forall\vy\phi(\vx;\vy)$,
where $\phi$ is a positive Boolean combination of formulas of the form
$\dom(x;y)$.%
\footnote{We focus on the distance-$1$ case here for simplicity,
    but all arguments work for the distance-$r$ case as well.}
The agreement graph $\phi(G)$ for a formula $\phi(\vx;\vy)$ is defined as a
bipartite graph over the vertex set $V(G)^{\Abs{\vx}}\cup V(G)^{\Abs{\vy}}$,
where for $\vc\in V(G)^{\Abs{\vx}}$ and $\vw\in V(G)^{\Abs{\vy}}$,
$(\vc,\vw)\in E(\phi(G))$ if and only if $G\models \phi(\vc;\vw)$.

The Semi-Ladder Algorithm is fpt for any domination-type problem on $G$ if the $\phi^k$-semi-ladder index of $G$ is bounded.
Similarly,
the Co-Matching Algorithm is fpt if the $\phi^k$-co-matching index of $G$
is bounded.
However,
while the former classes are always characterized by semi-ladder-freeness of
$G$ (\cref{cor:semi-ladder-free-delta_1^k} and \cite[Lemma 4]{FabianskiPST19}),
the latter sometimes collapses to semi-ladder-free classes
(as opposed to co-matching-free and double-ladder-free) as well,
in which case the Co-Matching Algorithm loses its advantage over the
Semi-Ladder Algorithm.

The reason is as follows.
The proof of \cref{lem:co-matching-Ramsey} is analogous to that of
\cite[Lemma~4]{FabianskiPST19},
where given a positive Boolean combination $\psi(\vx;\vy)$ of
formulas $\phi_1(\vx;\vy),\dots,\phi_k(\vx;\vy)$,
the semi-ladder index of $\psi(G)$ is bounded by the semi-ladder indices of
$\phi_1(G),\dots,\phi_k(G)$.
However,
due to the pair of colors in our proof,
it implicitly makes use of the
double-ladder index with respect to two formulas at once, where the two
ladders together form a co-matching.
In the following example, we show that
the double-ladder index with respect to two formulas obtained by adding
spurious variables to the same formula $\phi$ can in general
only be bounded by the
ladder index of $\phi(G)$ instead of its double-ladder index.

\begin{example}\label{ex:double-ladder-spurious} Let
    $\psi(\vx;\vy)\coloneq\phi_1(\vx;\vy)\vee\phi_2(\vx;\vy)$, where
    $\phi_1(\vx;\vy)\coloneq \dom(x^1;y^1)$ and $\phi_2(\vx;\vy)\coloneq
        \dom(x^2;y^2)$, and let $G$ contain a ladder $(a_1,b_1),\dots,(a_4,b_4)$
    of order $4$. Then, as demonstrated in \cref{fig:spurious},
    $((a_1,b_1),(b_1,a_1)),\dots,((a_4,b_4),(b_4,a_4))$ forms a co-matching
    of order $4$ in $\psi(G)$.
\end{example}
\begin{figure}[t]
    \centering
    \includegraphics{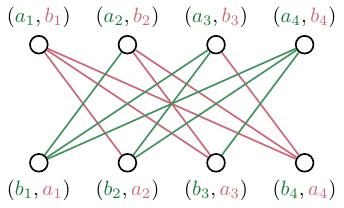}
    \caption{Illustration of a co-matching in $\psi(G)$ in
        \cref{ex:double-ladder-spurious}.
        $\phi_1$ is marked green and $\phi_2$ is marked red.}
    \label{fig:spurious}
\end{figure}
Essentially, the example shows that a ladder can be ``read'' twice in the
opposite direction to form a co-matching.
This example originates from \cite{MaehlmannS26}, in which co-matchings are
explored under a restricted notion of transduction. There, the example is
used to show that having a long ladder implies some form of
\emph{monadic non-equality property,}
which is similar to the $\phi$-co-matching index in
the progressive exploration framework.

To extend \cref{lem:co-matching-Ramsey} to allow all positive Boolean
combinations,
one would need to forbid adding spurious variables arbitrarily.
In particular,
one way to forbid reading a ladder twice is to fix $\Abs{\vy}=1$,
which is the case for $\dom_r^k(\vx;y)$.
Thus,
the issue does not appear for \textsc{(Distance-$r$) Dominating Set},
and \cref{lem:co-matching-Ramsey} can indeed be generalized to the
distance-$r$ setting.

Nonetheless,
we remark that
our approach can be applied to the following domination-type problems
proposed in \cite[Example~8]{FabianskiPST19}:
\begin{itemize}
    \item $y$ is at distance at most $r$ from at least two of the vertices
          $x^1,\dots,x^k$.
    \item The sum $\dist(x^1,y)+\dots+\dist(x^k,y)$ is at most $r$.
\end{itemize}

\subsection{Independent Set}
Next, we discuss the application of the progressive exploration framework
to the \textsc{Independent Set} problem.
There, they extend the algorithm to work on agreement graphs which are not
semi-ladder-free.
The agreement graph for the \textsc{Independent Set} problem has to have bounded ladder index and the \emph{weak $p$-Helly} property, a less restrictive version of the co-matching index.
This change is motivated by the fact that the agreement graph of many interesting graph classes, e.g. nowhere dense classes, is not semi-ladder-free for \textsc{Independent Set}.
The more restrictive algorithm comes at a cost:
The number of rounds the algorithm makes can only be
bounded by $R^p(2\ell)$, where $\ell$ is the ladder index of the agreement
graph.
This means that, even if the Compression Step can be added to the
Ladder Algorithm, this bound explodes to $O(|G|^p)$ for classes of
bounded $\dom^k$-co-matching index, instead of to $O(|G|)$.
Thus, we can only get an XP-time algorithm.

Moreover, it is recently shown that \textsc{Independent Set} is fpt in
semi-ladder-free graph classes using indiscernible
sequences~\cite{DreierMS26}, i.e., the weak Helly property is not
necessary. It seems that progressive
exploration may not be the best framework to tackle this problem.

\subsection{Partial Dominating Set}
\label{sec:pds}
We consider here the \textsc{Partial Dominating Set} problem,
which asks if there exists a set of size $k$ that dominates at least
$t\leq \Abs{G}$ vertices.
We show that we can adapt the progressive exploration framework to tackle
this variant,
albeit
the changes are of a different nature than the ones for \textsc{Distance-$r$
    Dominating Set} and only work for the Semi-Ladder Algorithm.

In contrast to the previous problems, \textsc{Partial Dominating Set} is not
expressible in first-order logic, especially not in the form
$\exists\vx\forall\vy\phi(\vx;\vy)$ for any formula $\phi$.

With a few simple changes, we can adapt the Semi-Ladder Algorithm to
\emph{approximate} the maximum value of $t$ of a \textsc{Partial Dominating
    Set} instance within a factor of $1-\epsilon$ on semi-ladder-free classes in
fpt time when $t = \Theta(\Abs{G})$ and for arbitrary $\epsilon>0$. Here, only
$k$, $\epsilon$ and the semi-ladder index of the agreement graph are
considered as parameters, not $t$.

The changes to the Semi-Ladder Algorithm are rather simple. The Witness
Oracle is replaced by the \emph{Randomized Witness Oracle} which selects
uniformly at random a vertex that does not agree with the current solution.
The analysis relies on the fact that if a large set is dominated,
then we can sample a witness from the set with sufficiently high probability.
Consequently, the algorithm only needs to be repeated a constant number of
times to
amplify the success probability to, say, $2/3$. The number of repetitions
depends only on $k$, $\epsilon$ and the semi-ladder index of the agreement
graph.

However, we refrain from giving more details here,
as the result is subsumed by the more general result
of \cite{BadanidiyuruKL12},
who establish that the optimal value of $t$ in the
\textsc{Partial Dominating Set} problem can be approximated within a factor
of $1-\epsilon$ in fpt time on graph classes of bounded VC-dimension.
This result holds for any $t$, also when $t = o(\Abs{G})$.
Nevertheless, it demonstrates the flexibility of the framework and may suggest potential for further adaptations.

\subsection{Set Cover}\label{sec:set-cover}
In \cite{Guillemot25},
Guillemot generalizes semi-ladder-free graph classes to semi-ladder-free
set systems (or, hypergraphs) and proposed a branching algorithm and a
kernelization for the \textsc{Set Cover} problem in such systems.
While the focus of \cite{Guillemot25} is to devise an algorithm that does
not rely on Ramsey arguments,
it provides an intuitive generalization of the structure and obstructions
into set systems.

In short,
a semi-ladder of order $\ell$ in a set system $(U,\F)$ is a sequence
$(u^1,S^1),\dots,(u^{\ell},S^{\ell})$ such that $u^i\in S^j\iff i <j$.
Since \textsc{Dominating Set} can be formulated as \textsc{Set Cover} using
the closed neighborhood set system,
their definition of semi-ladder-freeness corresponds to having
bounded $\dom$-semi-ladder index in the context of \textsc{Dominating Set}.

The \textsc{Set Cover} problem asks, given a finite universe $U$, a family
$\F\subseteq 2^U$, and an integer $k$, whether there exist at most $k$ sets
in $\F$ whose union is $U$.

We discuss here how our results apply in this setting.
First,
it is easy to see that all oracles can still be implemented efficiently in
the case of set systems.
By defining co-matchings and double-ladders in the same fashion,
the Ramsey argument relating $\dom^k$-semi-ladder indices to
$\dom$-semi-ladder indices~\cite[Lemma~4]{FabianskiPST19} and that relating
$\dom^k$-co-matching indices to $\dom$-co-matching indices and
$\dom$-double-ladder indices (\cref{lem:co-matching-Ramsey}) can be easily
adapted to work for $k$-fold unions of sets.
Thus,
the Co-Matching Algorithm (resp., the Semi-Ladder Algorithm) work for
\textsc{Set Cover} on co-matching-free
and double-ladder-free set systems (resp., on semi-ladder-free set systems).

\begin{theorem}
    \label{thm:set-cover-generalization}
    The Co-Matching Algorithm solves
    \textnormal{\textsc{Set Cover}} on
    co-matching-free and double-ladder-free set systems
    in fixed-parameter tractable time.
\end{theorem}

Moreover,
similar to how \textsc{Dominating Set} can be generalized to domination-type
problems,
as discussed in \cref{sec:domination-type},
we can generalize \textsc{Set Cover} to \emph{coverage-type problems}
by replacing the $k$-fold union by any formula $\phi(S^1,\dots,S^k)$
consisting of only unions and intersections on the $k$ sets.
The problem asks whether there exist sets $S^1,\dots,S^k\in\F$ such that
\[
    U\subseteq\phi(S^1,\dots,S^k).
\]
The discussion in \cref{sec:domination-type} applies to such problems,
and thus the Co-Matching Algorithm is fpt on co-matching-free
and double-ladder-free set systems for such problems.
We also remark that the branching algorithm proposed in \cite{Guillemot25}
seems to work for such problems as well with straightforward modification.

We remark that the open question posed in \cite{Guillemot25} whether ladder-free
\textsc{Set Cover} is fpt on ladder-free set systems can be answered negatively
using \cref{prop:domset-hardness},
since the closed neighborhood set system of a ladder-free class is also
ladder-free.
It will be interesting to study whether co-matching-free and double-ladder-free
set systems also contain exploitable structures that allow for more
efficient algorithms and/or kernelization.

\section{Conclusion}\label{sec:conclusion}

We have shown that for \textsc{Dominating Set}, the progressive exploration
framework can be made to work with only the Helly property, and that this is
equivalent to graph classes being both co-matching- and double-ladder-free. This
subsumes many known fixed-parameter tractable cases of \textsc{Dominating Set},
namely on biclique-free classes and weakly $\gamma$-closed classes while also
matching the asymptotic running time on those classes.

These results can be generalized to the distance-$r$ setting, and an
interesting direction is to characterize the classes of bounded $\dom_r^k$-co-matching
index, i.e., the classes where the Co-Matching Algorithm for \textsc{Distance-$r$ Dominating Set} is fpt.
It is unclear which well-studied classes have this property.

We are surprised that for bounded ladder index, the progressive exploration
framework does not (or cannot) work.
This answers a question stated by Guillemot~\cite{Guillemot25} in the negative.
Hence, we cannot generalize the
tractability results for \textsc{Dominating Set} on (monadically) stable
classes. Is there a notion which generalizes monadically stable classes but
still allows for fixed-parameter tractability of \textsc{Dominating Set}?

\bibliography{references.bib}

\end{document}